\documentclass[aps,pra,onecolumn,notitlepage,superscriptaddress]{revtex4-2}
\usepackage{epsfig,amssymb,amsmath,mathrsfs,amsfonts,amsthm,graphicx,float,color}
\usepackage[ruled,linesnumbered,procnumbered]{algorithm2e}
\usepackage{refcount,xr}
\usepackage[colorlinks=true]{hyperref}
\usepackage[capitalise]{cleveref}
\usepackage[active]{srcltx}
\usepackage{subfigure}
\usepackage{comment}
\usepackage{etoolbox}
\usepackage{pbox}

\makeatletter
\gdef\@ptsize{0} 
\makeatother
\usepackage{setspace}
\usepackage{fancyhdr}
\fancypagestyle{plain}{
\fancyfoot[C]{\thepage}}
 \def\map#1{\mathcal #1}
\def\<{\langle}\def\>{\rangle}

\def\Tr{\operatorname{Tr}}\def\:{\hbox{\bf
    :}}

\def\spc#1{\mathcal{#1}}
\def\rank{\mathsf{rank}}

\def\set#1{\mathsf{#1}}

\newtheorem{theo}{{Theorem}}

\crefname{theo}{Theorem}{Theorems}

\newtheorem{lem}{{Lemma}}

\crefname{lem}{Lemma}{Lemmas}

\newtheorem{prop}{{Proposition}}

\crefname{prop}{Proposition}{Propositions}

\newtheorem{cor}{{Corollary}}

\crefname{cor}{Corollary}{Corollaries}

\newtheorem{definition}{{Definition}}

\crefname{definition}{Definition}{Definitions}

\newtheorem{assumption}{{Assumption}}

\crefname{assumption}{Assumption}{Assumptions}

\SetAlgorithmName{Algorithm}{Algorithm}{List of algorithms}
\SetProcNameSty{texttt}
\SetAlgoFuncName{Function}{Function}

\makeatletter
\AtBeginEnvironment{function}{\let\c@algocf\c@function}
\makeatother

\makeatletter
\newcommand*{\storecounter}[2]{%
  \edef\@currentlabel{\the\value{#1}}
  \label{#2}
}
\makeatother

\newcommand{\ket}[1]{\left|#1\right\rangle}
\newcommand{\kket}[1]{\left|#1\middle\rangle\!\right\rangle}

\newcommand{\bra}[1]{\left\langle#1\right|}

\newcommand{\bbra}[1]{\left\langle\!\middle\langle#1\right|}

\newcommand{\braket}[2]{\left\langle#1\middle|#2\right\rangle}
\newcommand{\ketbra}[1]{\ket{#1}\bra{#1}}
\newcommand{\kketbbra}[1]{\kket{#1}\bbra{#1}}

\newcommand{\bbrakket}[2]{\left\langle\!\left\langle#1\middle|#2\right\rangle\!\right\rangle}

\newcommand{\ind}{\DataSty{ind}}
\newcommand{\True}{\KwSty{true}}
\newcommand{\False}{\KwSty{false}}
\newcommand{\independent}{\FuncSty{independence}}

\newcommand{\chione}[3]{\chi_1\!\left(#2;#3\right)_{#1}}

\usepackage{tikz}
\usetikzlibrary{calc}

\usetikzlibrary{decorations.pathreplacing, positioning, shapes.misc}
\tikzset{tensor/.style={rectangle,color=black,draw=black,fill=white,thick,
                    inner sep=1pt,minimum size=5mm}}
\tikzset{tensor2h/.style={rectangle,color=black,draw=black,fill=white,thick,
                    inner sep=1pt,minimum width=5mm,minimum height=10mm}}
\tikzset{parameter/.style={rectangle,color=black,draw=black,fill=black!10,thick,
                    inner sep=1pt,minimum size=5mm}}
\tikzset{virtual/.style={rectangle,inner sep=1pt,minimum size=5mm}}
\tikzset{prepare/.style={rounded rectangle, rounded rectangle east arc=none,color=black,draw=black,fill=white,thick,inner sep=1pt,minimum size=5mm}}
\tikzset{measure/.style={rounded rectangle, rounded rectangle west arc=none,color=black,draw=black,fill=white,thick,inner sep=1pt,minimum size=5mm}}
\tikzset{
    triple3/.style args={[#1] in [#2] in [#3]}{
        #1,preaction={preaction={draw,#3},draw,#2}
    }
}
\tikzset{triple/.style={triple3={[line width=0.125mm,black] in [line width=2mm,white] in [line width=2.25mm,black]}}}
\tikzset{thick triple/.style={triple3={[line width=0.25mm,black] in [line width=1.75mm,white] in [line width=2.25mm,black]}}}

\newcommand{\measurement}{\draw (55:0.5) arc (55:125:0.5); \draw (80:0.25) -- (80:0.6);}
\newcommand{\drawground}{\draw[thick] (-0.15,0) -- (0.15,0);\draw[thick] (-0.10,-0.05) -- (0.10,-0.05);\draw[thick] (-0.05,-0.1) -- (0.05,-0.1);}
\tikzset{ground/.pic={\drawground}}

\newenvironment{mathtikz}[1][]{\begin{array}{c}\begin{tikzpicture}[#1]}{\end{tikzpicture}\end{array}}

\makeatletter

\pgfkeys{/pgf/.cd,
  rectangle corner radius north west/.initial=10pt,
  rectangle corner radius north east/.initial=10pt,
  rectangle corner radius south west/.initial=10pt,
  rectangle corner radius south east/.initial=10pt
}
\newif\ifpgf@rectanglewrc@donecorner@
\def\pgf@rectanglewithroundedcorners@docorner#1#2#3#4#5{%
  \edef\pgf@marshal{%
    \noexpand\pgfintersectionofpaths
      {%
        \noexpand\pgfpathmoveto{\noexpand\pgfpoint{\the\pgf@xa}{\the\pgf@ya}}%
        \noexpand\pgfpathlineto{\noexpand\pgfpoint{\the\pgf@x}{\the\pgf@y}}%
      }%
      {%
        \noexpand\pgfpathmoveto{\noexpand\pgfpointadd
          {\noexpand\pgfpoint{\the\pgf@xc}{\the\pgf@yc}}%
          {\noexpand\pgfpoint{#1}{#2}}}%
        \noexpand\pgfpatharc{#3}{#4}{#5}%
      }%
    }%
  \pgf@process{\pgf@marshal\pgfpointintersectionsolution{1}}%
  \pgf@process{\pgftransforminvert\pgfpointtransformed{}}%
  \pgf@rectanglewrc@donecorner@true
}
\pgfdeclareshape{rectangle with rounded corners}
{
  \inheritsavedanchors[from=rectangle] 
  \inheritanchor[from=rectangle]{north}
  \inheritanchor[from=rectangle]{north west}
  \inheritanchor[from=rectangle]{north east}
  \inheritanchor[from=rectangle]{center}
  \inheritanchor[from=rectangle]{west}
  \inheritanchor[from=rectangle]{east}
  \inheritanchor[from=rectangle]{mid}
  \inheritanchor[from=rectangle]{mid west}
  \inheritanchor[from=rectangle]{mid east}
  \inheritanchor[from=rectangle]{base}
  \inheritanchor[from=rectangle]{base west}
  \inheritanchor[from=rectangle]{base east}
  \inheritanchor[from=rectangle]{south}
  \inheritanchor[from=rectangle]{south west}
  \inheritanchor[from=rectangle]{south east}

  \savedmacro\cornerradiusnw{%
    \edef\cornerradiusnw{\pgfkeysvalueof{/pgf/rectangle corner radius north west}}%
  }
  \savedmacro\cornerradiusne{%
    \edef\cornerradiusne{\pgfkeysvalueof{/pgf/rectangle corner radius north east}}%
  }
  \savedmacro\cornerradiussw{%
    \edef\cornerradiussw{\pgfkeysvalueof{/pgf/rectangle corner radius south west}}%
  }
  \savedmacro\cornerradiusse{%
    \edef\cornerradiusse{\pgfkeysvalueof{/pgf/rectangle corner radius south east}}%
  }

  \backgroundpath{%
    \northeast\advance\pgf@y-\cornerradiusne\relax
    \pgfpathmoveto{}%
    \pgfpatharc{0}{90}{\cornerradiusne}%
    \northeast\pgf@ya=\pgf@y\southwest\advance\pgf@x\cornerradiusnw\relax\pgf@y=\pgf@ya
    \pgfpathlineto{}%
    \pgfpatharc{90}{180}{\cornerradiusnw}%
    \southwest\advance\pgf@y\cornerradiussw\relax
    \pgfpathlineto{}%
    \pgfpatharc{180}{270}{\cornerradiussw}%
    \northeast\pgf@xa=\pgf@x\advance\pgf@xa-\cornerradiusse\southwest\pgf@x=\pgf@xa
    \pgfpathlineto{}%
    \pgfpatharc{270}{360}{\cornerradiusse}%
    \northeast\advance\pgf@y-\cornerradiusne\relax
    \pgfpathlineto{}%
    \pgfpathclose
  }

  \anchor{before north east}{\northeast\advance\pgf@y-\cornerradiusne}
  \anchor{after north east}{\northeast\advance\pgf@x-\cornerradiusne}
  \anchor{before north west}{\southwest\pgf@xa=\pgf@x\advance\pgf@xa\cornerradiusnw
    \northeast\pgf@x=\pgf@xa}
  \anchor{after north west}{\northeast\pgf@ya=\pgf@y\advance\pgf@ya-\cornerradiusnw
    \southwest\pgf@y=\pgf@ya}
  \anchor{before south west}{\southwest\advance\pgf@y\cornerradiussw}
  \anchor{after south west}{\southwest\advance\pgf@x\cornerradiussw}
  \anchor{before south east}{\northeast\pgf@xa=\pgf@x\advance\pgf@xa-\cornerradiusse
    \southwest\pgf@x=\pgf@xa}
  \anchor{after south east}{\southwest\pgf@ya=\pgf@y\advance\pgf@ya\cornerradiusse
    \northeast\pgf@y=\pgf@ya}

  \anchorborder{%
    \pgf@xb=\pgf@x
    \pgf@yb=\pgf@y%
    \southwest%
    \pgf@xa=\pgf@x
    \pgf@ya=\pgf@y%
    \northeast%
    \advance\pgf@x by-\pgf@xa%
    \advance\pgf@y by-\pgf@ya%
    \pgf@xc=.5\pgf@x
    \pgf@yc=.5\pgf@y%
    \advance\pgf@xa by\pgf@xc
    \advance\pgf@ya by\pgf@yc%
    \edef\pgf@marshal{%
      \noexpand\pgfpointborderrectangle
      {\noexpand\pgfqpoint{\the\pgf@xb}{\the\pgf@yb}}
      {\noexpand\pgfqpoint{\the\pgf@xc}{\the\pgf@yc}}%
    }%
    \pgf@process{\pgf@marshal}%
    \advance\pgf@x by\pgf@xa%
    \advance\pgf@y by\pgf@ya%
    \pgfextract@process\borderpoint{}%
    \pgf@rectanglewrc@donecorner@false
    \southwest\pgf@xc=\pgf@x\pgf@yc=\pgf@y
    \advance\pgf@xc\cornerradiussw\relax\advance\pgf@yc\cornerradiussw\relax
    \borderpoint
    \ifdim\pgf@x<\pgf@xc\relax\ifdim\pgf@y<\pgf@yc\relax
      \pgf@rectanglewithroundedcorners@docorner{-\cornerradiussw}{0pt}{180}{270}{\cornerradiussw}%
    \fi\fi
    \ifpgf@rectanglewrc@donecorner@\else
      \southwest\pgf@yc=\pgf@y\relax\northeast\pgf@xc=\pgf@x\relax
      \advance\pgf@xc-\cornerradiusse\relax\advance\pgf@yc\cornerradiusse\relax
      \borderpoint
      \ifdim\pgf@x>\pgf@xc\relax\ifdim\pgf@y<\pgf@yc\relax
       \pgf@rectanglewithroundedcorners@docorner{0pt}{-\cornerradiusse}{270}{360}{\cornerradiusse}%
      \fi\fi
    \fi
    \ifpgf@rectanglewrc@donecorner@\else
      \northeast\pgf@xc=\pgf@x\relax\pgf@yc=\pgf@y\relax
      \advance\pgf@xc-\cornerradiusne\relax\advance\pgf@yc-\cornerradiusne\relax
      \borderpoint
      \ifdim\pgf@x>\pgf@xc\relax\ifdim\pgf@y>\pgf@yc\relax
       \pgf@rectanglewithroundedcorners@docorner{\cornerradiusne}{0pt}{0}{90}{\cornerradiusne}%
      \fi\fi
    \fi
    \ifpgf@rectanglewrc@donecorner@\else
      \northeast\pgf@yc=\pgf@y\relax\southwest\pgf@xc=\pgf@x\relax
      \advance\pgf@xc\cornerradiusnw\relax\advance\pgf@yc-\cornerradiusnw\relax
      \borderpoint
      \ifdim\pgf@x<\pgf@xc\relax\ifdim\pgf@y>\pgf@yc\relax
       \pgf@rectanglewithroundedcorners@docorner{0pt}{\cornerradiusnw}{90}{180}{\cornerradiusnw}%
      \fi\fi
    \fi
  }
}

\makeatother

\tikzset{prepare tall/.style={rectangle with rounded corners, rectangle corner radius north east=0pt, rectangle corner radius south east=0pt, rectangle corner radius north west=30pt, rectangle corner radius south west=30pt, color=black,draw=black,fill=white,thick,inner sep=1pt,minimum height=22mm,minimum width=12mm}}
\tikzset{measure tall/.style={rectangle with rounded corners, rectangle corner radius north west=0pt, rectangle corner radius south west=0pt, rectangle corner radius north east=30pt, rectangle corner radius south east=30pt, color=black,draw=black,fill=white,thick,inner sep=1pt,minimum height=22mm,minimum width=12mm}}

\usepackage{array} 
\usepackage{etoolbox}

\newcommand{\red}[1]{{#1}}

\newcommand{\supplementref}[1]{\hyperref[#1]{Supplementary Note\hspace{0.2em}\ref*{#1}}}

\newif\iffigures
\figurestrue
\graphicspath{{figures/}} 

\usepackage{titlesec}

\begin{document}

\title{
Quantum Causal Unravelling}
\author{Ge Bai}
\affiliation{QICI Quantum Information and Computation Initiative, Department of Computer Science, The University of Hong Kong, Pokfulam Road, Hong Kong}
\affiliation{HKU-Oxford Joint Laboratory for Quantum Information and Computation}
\author{Ya-Dong Wu}
\affiliation{QICI Quantum Information and Computation Initiative, Department of Computer Science, The University of Hong Kong, Pokfulam Road, Hong Kong}
\affiliation{HKU-Oxford Joint Laboratory for Quantum Information and Computation}
\author{Yan Zhu}
\affiliation{QICI Quantum Information and Computation Initiative, Department of Computer Science, The University of Hong Kong, Pokfulam Road, Hong Kong}
\affiliation{HKU-Oxford Joint Laboratory for Quantum Information and Computation}
\author{Masahito Hayashi}
\affiliation{Shenzhen Institute for Quantum Science and Engineering, Southern University of Science and Technology, Shenzhen 518055, China \looseness=-1}
\affiliation{Guangdong Provincial Key Laboratory of Quantum Science and Engineering, Southern University of Science and Technology, Shenzhen 518055, China}
\affiliation{Graduate School of Mathematics, Nagoya University, Nagoya, 464-8602, Japan}
\author{Giulio Chiribella}
\email{giulio@cs.hku.hk}
\affiliation{QICI Quantum Information and Computation Initiative, Department of Computer Science, The University of Hong Kong, Pokfulam Road, Hong Kong}
\affiliation{Department of Computer Science, University of Oxford, Parks Road, Oxford OX1 3QD, United Kingdom\looseness=-1}
\affiliation{HKU-Oxford Joint Laboratory for Quantum Information and Computation}
\affiliation{Perimeter Institute For Theoretical Physics, 31 Caroline Street North, Waterloo N2L 2Y5, Ontario, Canada\looseness=-1}

\begin{abstract}
Complex  processes often arise from sequences of simpler  interactions involving a few particles at a time. 
 These interactions, however, may not be directly accessible to experiments.
    Here we develop the first efficient  method  for unravelling the causal structure of the  interactions in a multipartite quantum process, under the assumption that the  process has bounded information loss and induces  causal dependencies whose strength is  above a  fixed (but otherwise arbitrary)  threshold. Our method is based on a quantum algorithm whose complexity  scales  polynomially in the total number of input/output systems, in the dimension of the systems involved in each interaction,  and in the inverse of the chosen threshold for the strength of the causal dependencies. 
Under additional assumptions, 
we also provide a second algorithm   that has lower complexity and  requires only local state preparation and local measurements.  
Our algorithms can be used to  identify processes that can be characterized efficiently with the technique of quantum process tomography.  Similarly, they can be used to identify useful communication channels in quantum networks, and to test  the internal structure  of uncharacterized quantum circuits. 
\end{abstract}


\maketitle

\section{Introduction}

Many  processes in nature arise from sequences of basic interactions, each involving a small number of physical systems. Determining  the causal structure of these interactions is important  both for  basic science and for engineering.    Often, however, the sequence of interactions giving rise to  a process of  interest may  not be directly accessible to experiments. 
For example, scattering experiments in high energy physics can probe the  relation between a set of incoming particles  and a set of outgoing particles,  but typically cannot access the individual events taking place within the scattering region. 
  In this and similar scenarios, a   fundamental problem  is to characterize the causal  structure of the interactions  by accessing only  the inputs and outputs of the process of interest, while treating the intermediate steps as a black box.    
 We call this problem, illustrated in  \autoref{fig:title},  the causal unravelling of an unknown physical process.  Explicitly, the problem of causal unravelling  is  to determine whether an unknown  process can be broken down into a sequence of simpler interactions, to determine the order of such interactions, and to determine which systems  take part in each interaction. Causal unravelling can be viewed as a special case of the broader problem of causal discovery \cite{spirtes2000causation,pearl2009causality}, namely  the task  to identify the causal relations between a given set of variables.   In the broad class of causal discovery problems, the  distinctive features of causal unraveling are that {\em (i)}  the goal is to identify a linear causal structure, corresponding to the sequence of interactions underlying the given process, and {\em (ii)}  certain variables are {\em a priori} known to be `inputs' (and therefore potential `causes'), while other variables are {\em a priori} known to be `outputs'  (and therefore potential `effects'). This scenario often arises in experimental physics,  where the input/output structure is typically clear from the design of the experiment, as in the aforementioned example of scattering experiments.
 In principle,  candidate answers  can be extracted from a full tomographic characterization of the process under consideration. However, the complexity of  process tomography grows exponentially in the number of inputs and outputs, making this approach unfeasible when the process involves a large number of systems.    

 In the classical domain, the problem of causal unravelling can be  efficiently addressed  with a variety of   algorithms developed for the general problem of causal discovery \cite{spirtes2000causation,pearl2009causality,heinze2018causal}.  Classical causal discovery algorithms often formulate causal relationships with graphical models and solve such models by structure learning algorithms, such as PC algorithm (named after Peter Spirtes and Clark Glymour) \cite{spirtes2000causation}, Greedy Equivalence Search \cite{chickering2002optimal}, and Max-Min Hill Climbing \cite{tsamardinos2006max}. These algorithms cover a wide variety of problems, by making different sets of assumptions on   the process under consideration.  Typical assumptions include  causal sufficiency---meaning that no variables are hidden---and causal faithfulness---meaning that the conditional independences among the variables are precisely those  associated to an underlying  graph used to model the causal structure.   In general, however, causal discovery is intrinsically a hard problem: when no assumption is made,  the complexity of all the known algorithms becomes exponential in the worst case over all possible instances \cite{chickering1996learning,chickering2004large}.

In the quantum domain, the problem of causal unravelling  is made even more challenging by the presence of  correlations that elude a classical  explanation~\cite{wood2015lesson,van2019quantum}. 
  In recent years, the  quantum  extension of the notion of causal model  has been addressed in a series  of works \cite{henson2014theory,pienaar2015graph,costa2016quantum,allen2017quantum,barrett2019quantum,barrett2021cyclic}, 
 providing a solid conceptual foundation to the field of quantum causal discovery. 
On the algorithmic side, however, the study of quantum causal models remained relatively underdeveloped.  Specific instances of quantum causal discovery  were studied in Refs. \cite{ried2015quantum,fitzsimons2015quantum,chiribella2019quantum}, showing that quantum resources  offer appealing advantages. 
  These examples, however, were limited to simple instances, typically involving a small number of variables and/or a small number of  hypotheses on the  causal structure.  In more general scenarios, one approach could be to perform  quantum process tomography 
  and then to infer the causal structure from the full description of the process under consideration  \cite{giarmatzi2018quantum}.  As in the classical case,  however,  the number of queries needed by  a full process tomography grows exponentially with the number of systems involved in the process, making this approach impractical as the size of the problem increases. 

\begin{figure}[h]
\centering
\iffigures
  \includegraphics[width=0.8\linewidth]{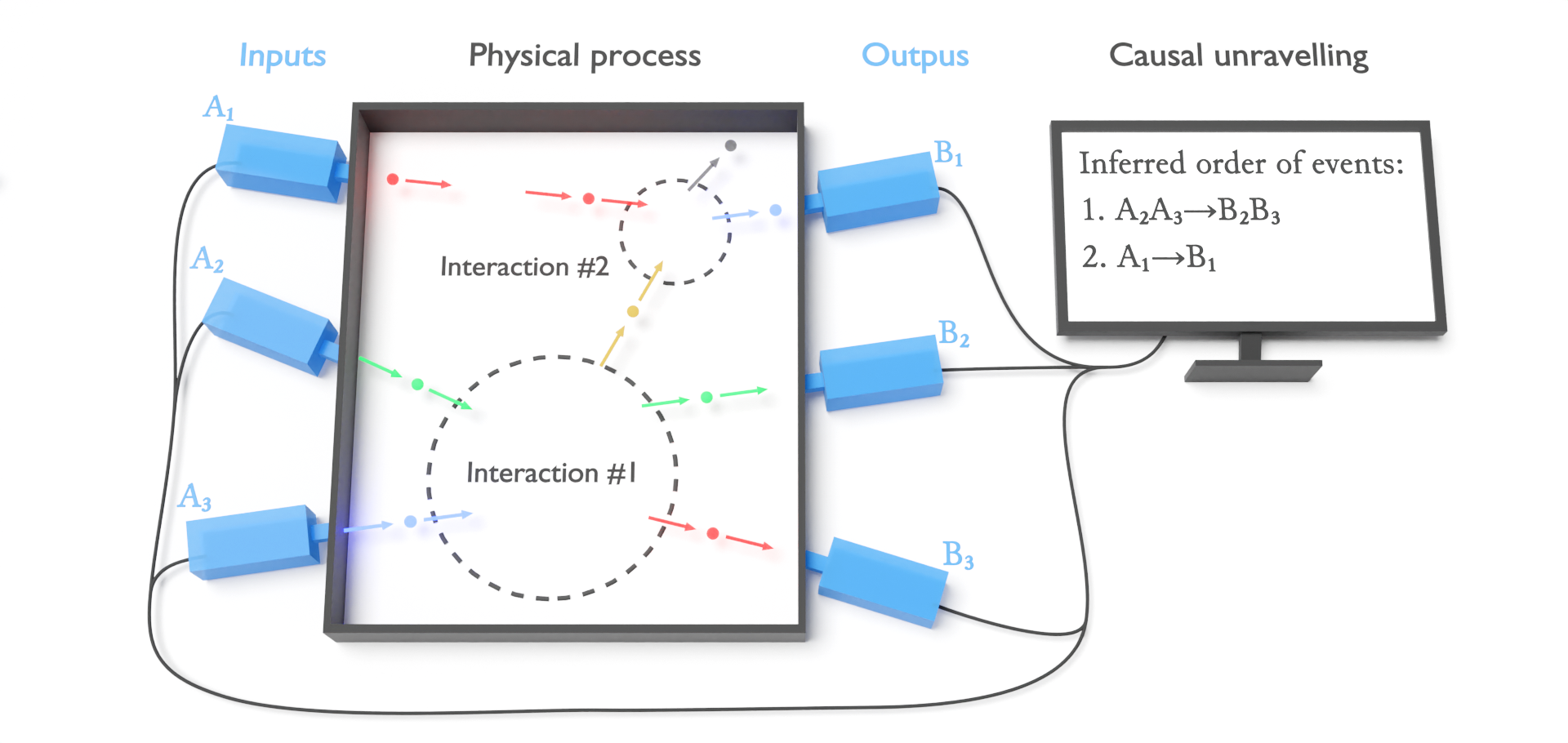}
\fi
\caption{\label{fig:title}  {\bf Causal unravelling.}  A physical process involves a set of input systems  (in the figure, $A_1, A_2,$ and $A_3$) and a set of output systems  (in the figure, $B_1, B_2,$ and  $B_3$).  The process is the result of a sequence of interactions, each of which involves a subset of the inputs and a subset of the outputs. The problem is to infer  the causal structure  of the  interactions solely from the input-output behaviour of the process. }
\end{figure}

In this paper, we provide an efficient algorithm for unravelling the causal structure of multipartite quantum processes without resorting to  full process tomography.  Our algorithm is similar to the PC algorithm \cite{spirtes2000causation} for classical causal discovery, in that it is based on a set of tests that establish the independence relations between subsets of input and output systems. We show that the algorithm has the following features: 
\begin{enumerate}
  \item The number of independence tests needed to infer the causal structure scales polynomially with the number of inputs/outputs of the process.   This feature is possible thanks to the special structure of the causal unravelling problem, where the goal is to establish a linear ordering of the interactions giving rise to the process under consideration.
  \item The independence tests produce, as a byproduct,  an estimate of the strength of correlation between the various inputs and outputs of the process.  In Methods, we show that this estimate can be obtained by performing a number of measurements that grows polynomially with the dimension of the systems under consideration, and that the number of measurements needed to conclude independence scales polynomially with the number of systems.
  \item The algorithm is exact whenever the process has bounded information loss, and satisfies a form of causal faithfulness property, namely that the strength of the causal relations, when present, is above a given threshold. When these assumptions are not satisfied, the algorithm produces an approximate result. The details of the approximate case are in \supplementref{app:approximate}.
\end{enumerate}

Moreover, the efficiency of our algorithm can be further boosted in special cases, including {\em(i)} the case where each input of a given interaction has a non-trivial causal influence on all the outputs of subsequent interactions, and {\em(ii)} the case where the process belongs to a special case of Markovian processes \cite{costa2016quantum,giarmatzi2018quantum,pollock2018operational,berk2021resource}, where each output depends only on one previous input and each input affects only one later output.  
We study these cases in the Results, where we devise an alternative algorithm that only requires local state preparations and local measurements. The number of queries to the process is only logarithmic in the number of input and output wires, thanks to a method that efficiently determines the correlations between input-output pairs as described in the Methods.

The remaining parts of this paper is structured as follows. In Results, we first formulate the quantum causal unravelling problem. In the second subsection, we give the main body of the efficient causal unravelling algorithm, discuss the assumptions and analyze its efficiency. The third subsection of Results talks about the alternative algorithm designed for special cases. At the end of Results, we briefly talk about a generalization of our algorithm. 
Future works, interpretations and the applications of the causal unravelling algorithms are addressed in the Discussion. The Methods section contains the detailed implementation of the independence tests used by the algorithms in Results.

\section{Results}


{\bf Problem formulation.}  Let us start by giving a precise formulation of the problem of quantum causal unravelling.  In this problem, an experimenter is given access to a multipartite  quantum process, with inputs labelled as $A_1,\dots,A_{n_{\rm in}}$, and outputs labelled as $B_1,\dots,B_{n_{\rm out}}$.  Note that, in general, different labels may refer to the same physical system: for example, system $A_1$ could be a single photon with a given frequency, entering in the interaction region,   and system $B_1$ could be a single photon with the same frequency, exiting the interaction region.    Mathematically,  the   process is described by a quantum channel, that is,  a completely positive trace-preserving (CPTP) linear map   $\map{C}$ transforming operators on the tensor product space $\spc{H}_{A_1}\otimes\dots\otimes\spc{H}_{A_{n_{\rm in}}}$ to operators on the tensor product space $\spc{H}_{B_1}\otimes\dots\otimes\spc{H}_{B_{n_{\rm out}}}$.    Note that, without loss of generality, one can always assume $n_{\rm in} = n_{\rm out} = n$,  as this condition can be satisfied  by adding a number of dummy systems with one-dimensional Hilbert space. In the following, we will denote by $L(\spc H)$  the set of linear operators on a generic Hilbert space $\spc H$, and by  $S(\spc H)$  the subset of density operators on $\spc H$, that is, the subset of operators $\rho  \in  L(\spc H)$  that are positive semidefinite and have unit trace.  

 \begin{figure}[h]
\iffigures
\includegraphics{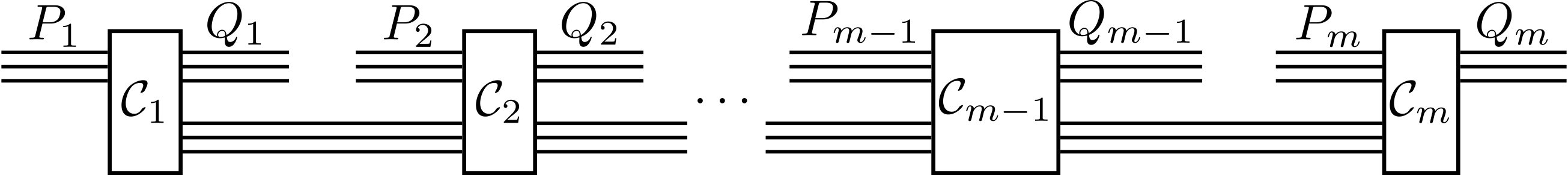}
\fi
\caption{\label{fig:PQcomb}  {\bf Decomposition of a  physical process into a  sequence of  $m$ interactions.} The inputs  (outputs) of the processes are divided into $m$ non-overlapping subsets.  The $i$-th subset of the inputs (outputs) consists of the systems that enter (exit)  the interaction at the $i$-th step (rectangular boxes in the picture). The wires connecting one interaction to the next represent intermediate systems that in this stage are not directly accessible to experiments.  For example, a photon could  enter the first interaction, remain as an intermediate system  between the first and second interaction, and then  exit the interaction region after the second interaction.  }
\end{figure}

The problem of causal unravelling is to determine whether a multipartite process can be broken down into a sequence of interactions, as in \autoref{fig:PQcomb}, and, in the affirmative case, to determine which systems are involved in each interaction.  
Mathematically, the problem is to find a partition  $\{P_i\}_{i=1}^m$  of the set $\{A_1,\dots,A_n\}$ and a partition $\{Q_i\}_{i=1}^m$ of the set $\{B_1,\dots,B_n\}$, such that the multipartite process can be decomposed into a sequence of interactions, with the $i$-th interaction involving input systems in $P_i$ and output systems in $Q_i$.  Such a sequential structure matches the framework of quantum combs \cite{chiribella2008quantum,chiribella2009theoretical}. A quantum comb is a quantum process that can be broken down into a sequence of interactions $\map{C}_1,\dots,\map{C}_k$ as in \autoref{fig:PQcomb}, while each interaction $\map{C}_i$ is a CPTP map and is called a tooth of the comb. Refs. \cite{chiribella2008quantum,chiribella2009theoretical} give a set of necessary and sufficient conditions for determining whether a given process conforms to a quantum comb, which is equivalent to whether the process admits a causal unravelling with partitions $\{P_i\}$ and $\{Q_i\}$.
We say a process $\map{C}$ has a causal unravelling $(P_1,Q_1),\dots,(P_m,Q_m)$ if it can be decomposed into the form of a quantum comb with $m$ teeth as  in \autoref{fig:PQcomb}.

\begin{figure}[h]
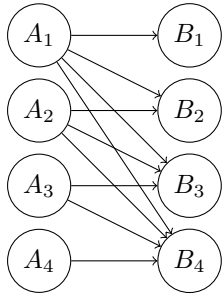

\iffigures
$$\begin{mathtikz}
    \foreach \i in {1,2,3,4} {
        \node[circle,draw=black] (A\i) at (0,-\i) {$A_{\i}$};
        \node[circle,draw=black] (B\i) at (2,-\i) {$B_{\i}$};
    }
    \foreach \i in {1,2,3,4}
        \foreach \j in {\i,...,4}
            \draw[->] (A\i) -- (B\j);
\end{mathtikz}$$
\fi
\caption{\label{fig:bipartite_graph}{\bf Causal graph of a quantum process with causal unravelling $(A_1,B_1),(A_2,B_2),(A_3,B_3),(A_4,B_4)$.} This is a bipartite graph because inputs can be independently controlled, and every two outputs are independent conditioned on all inputs. }
\end{figure}

The causal unravelling of a given multipartite process  determines the possible signalling relations between inputs and outputs.  With respect to this decomposition,  input systems at a given time  can only signal to output systems at later times.  The resulting   pattern can be graphically illustrated by a causal graph \cite{spirtes2000causation,pearl2009causality,barrett2019quantum},   as shown in \autoref{fig:bipartite_graph}.   Note that the causal graph includes all the signalling relations that are in principle compatible with the structure of the interactions. However,   a specific quantum  process may not exhibit any signalling  from a specific input to a specific output, even though the causal structure of the interactions would in principle permit it. When this occurs, the causal unravelling of a process may not be unique. For example, consider a bipartite process with inputs $\{A_1, A_2\}$ and outputs $\{B_1,  B_2\}$, with the property that $A_1$ signals to $B_1$ but not to $B_2$ and $A_2$ signals to $B_2$ but not to $B_1$. It is impossible to decide whether the signalling $A_1\to B_1$ happens before or after $A_2 \to B_2$ since they are causally uncorrelated. Therefore, the process admits two causal unravellings: $(A_1,B_1),(A_2,B_2)$ and $(A_2,B_2),(A_1,B_1)$. In this case, any of the two options is a valid solution of the causal unravelling problem.
   
Generally, causal discovery is a hard problem, even in the classical setting \cite{chickering1996learning,chickering2004large}. However, we will now show that, under a few assumptions, the more specific problem of quantum causal unravelling defined above can be solved efficiently. The first  assumption is that the basic interactions appearing in the causal unravelling involve a small number of systems, independent of the number of inputs and outputs.  At the fundamental level, this assumption is motivated by the fact that interactions are local, and typically involve a small number of systems.
 Mathematically, the assumption is that the cardinality of all  the sets   in the  partitions $\{   P_i\}_{i=1}^m$  and  $\{Q_i\}_{i=1}^m$ is no larger than a constant $c$ independent of  $n$.    For simplicity,  we will first restrict our attention to the special case   $c=1$,   meaning that the process can be broken down into interactions involving only one input and one output at a time. We will discuss larger $c$ at the end of the Results and in \supplementref{app:largerc}. 
    Hereafter, we will denote by  $\set{Comb}[(A_1,B_1),\dots,(A_n,B_n)]$ the set of quantum combs with $n$ teeth where the $i$-th tooth has input $A_i$ and output $B_i$.

In the basic $c=1$ scenario illustrated above, the problem of causal unravelling is to find out which pair of systems is involved in the first interaction, which pair is  involved in the second, and so on.   Formally, our goal is to identify an ordering of the inputs and outputs, $(A_{\sigma(1)},B_{\pi(1)}),\dots,(A_{\sigma(n)},B_{\pi(n)})$ with $\sigma$ and $\pi$ being permutations of $\{1,\dots,n\}$, such that 
$\map{C} \in \set{Comb}[(A_{\sigma(1)},B_{\pi(1)}),\dots,(A_{\sigma(n)},B_{\pi(n)})]$.

To quantify the efficiency of our algorithms, we will focus on the sample complexity, namely the number of black-box queries to the channel $\map{C}$, and on the computational complexity, including additional quantum and classical computation time measured by the number of elementary quantum gates and classical operations.


\medskip 

{\bf Efficient quantum  causal unravelling.}  We now provide an efficient quantum algorithm for quantum causal unravelling.   The main idea of the algorithm is to recursively find the last interaction in the decomposition of a given process.   
To illustrate this idea, consider the example of  \autoref{fig:bipartite_graph}. When we remove $B_4$, the node $A_4$ becomes disconnected from all the other nodes, meaning that the state of system $A_4$ does not affect the state of  the other systems.   By testing this independence condition, we can in principle check whether the graph of  \autoref{fig:bipartite_graph}    is an appropriate model for the process under consideration.  
 In general, 
suppose $(A_x,B_y)$ are the input and output of the last interaction. Since $A_x$ signals to only $B_y$, if we ignore $B_y$, $A_x$ is independent of the joint system formed by all systems excluding $A_x$ and $B_y$. This gives the following criterion that the last interaction must satisfy: 

\begin{prop} \label{prop:last_tooth_informal} \cite{chiribella2008quantum}
    The last interaction of a process $\map{C}$ involves the input/output pair $(A_x,B_y)$  if and only if $A_x$ is independent of $A_{\neq x}B_{\neq y}$,  the joint system containing all systems other than $A_x$ and $B_y$. 
\end{prop}
More details can be found in \supplementref{app:pre}.   By scanning the possible pairs $(A_x,B_y)$, we can find if one of them satisfy  the above criterion, and in the affirmative case, we can assign that pair to  the last interaction. 
Note that, in general, there may be more than one pair that satisfy the required condition. After one pair is found, one can reduce the problem to a smaller graph   containing $n-1$ inputs and $n-1$ outputs.   If at every step a suitable pair is found, then the final result is a valid causal unravelling of the original process. If at one step no pair can be found, the algorithm will then conclude that no causal unravelling with $c=1$ exists for the remaining subgraph. At this point, the algorithm can continue by considering causal unravellings with higher values of $c$ for the remaining subgraph, which is discussed at the end of the Results and detailed in \supplementref{app:largerc}.

A description of the algorithm is  provided  in \autoref{func:recursive}. The algorithm runs in a recursive manner: for an $n$-tooth comb, it finds the last tooth of the comb, remove the tooth (feeds the input  with an arbitrary state and discards its output), and reduces the problem to finding the causal unravelling of an $(n-1)$-tooth comb. 
We repeat the above procedure until we reach the bottom case $n=1$,  and thus obtain the order of all inputs and outputs. If the last tooth cannot be found in some iteration, it means that the current channel cannot be further decomposed, and the algorithm will output the trivial causal unravelling $(P,Q)$ where $P$ ($Q$) is the set of all input (output) wires of the current channel.

\begin{algorithm}[H]
    \caption{Efficient quantum causal unravelling algorithm}\label{func:recursive}
    \Indentp{0.5em}
    \SetInd{0.5em}{1em}
    \SetKwInOut{Preprocessing}{Preprocessing}
    \SetKwInOut{Input}{Input}
    \SetKwInOut{Output}{Output}
    \SetKw{Next}{next}
    \SetKwFunction{unravel}{unravel}

    \ResetInOut{Input}
    \Input{Black-box access to quantum channel $\map{C}$}
    \ResetInOut{Output}
    \Output{A causal unravelling  of $\map{C}$}
    \BlankLine
    \If{$\map{C}$ has only one input wire $A_x$ and one output wire $B_y$}
    {
        Output $(A_x,B_y)$ and exit\;
    }
    \ForEach{input-output pair $(A_x,B_y)$ of $\map{C}$}
    {
        \If(\tcp*[f]{Done by checking whether $A_x$ and $A_{\neq x}B_{\neq y}$ are independent using the quantum circuit described in the Methods}){$(A_x,B_y)$ is the last tooth\label{line:ind}}
        {
            Let $\map{C}_{A_{\neq x},B_{\neq y}}(\rho) := \Tr_{B_y}[\map{C}(\rho \otimes \tau_{A_x})]$, where $\tau_{A_x}$ is an arbitrary state on system $A_x$ \;
            \DontPrintSemicolon\tcp*[r]{Remove $A_x$ and $B_y$ from $\map{C}$}\PrintSemicolon
            Recursively run this algorithm on the reduced channel $\map{C}_{A_{\neq x},B_{\neq y}}$\;
            Output $(A_x,B_y)$ and exit\tcp*{Append $(A_x,B_y)$ to the end of the output}
        }
    }
    Let $P$ ($Q$) be the set of all input (output) wires of $\map{C}$\;
    Output $(P,Q)$ and exit\tcp*{Cannot be further decomposed}
\end{algorithm}

We now discuss the efficiency of this algorithm. First, we show that the number of independence tests is polynomial in $n$. 
Let $T_{\rm test}(n)$ be the number of independence tests required for an $n$-to-$n$ channel. In this algorithm, an independence test is performed for at most each of the $n^2$ input-output pairs $(A_x,B_y)$, resulting in $O(n^2)$ independence tests. After finding a last tooth, the problem size is reduced to $n-1$. Therefore, $T_{\rm test}(n)$ can be given by the recursive relation $T_{\rm test}(n) = O(n^2)+T_{\rm test}(n-1)$ with $T_{\rm test}(1)=O(1)$, solving which gives $T_{\rm test}(n)=O(n^3)$. 

Second, we show that the independence tests can be efficiently realized. This is a non-trivial problem, because testing whether two generic systems are in a product state is computationally hard in the worst-case scenario \cite{gutoski2013quantum}. 
Nevertheless, in the Methods we design a quantum circuit that performs the independence tests efficiently under assumptions of bounded information loss and causal faithfulness. Our circuit converts the independence test to the estimation of the distance between quantum states, which is done by the SWAP test \cite{buhrman2001quantum}.

Now, we give the efficiency guarantee for our algorithm. We first discuss the exact case, when our algorithm produces the exact causal unravelling of $\map{C}$ based on two assumptions. We will use a few parameters related to the Choi state \cite{choi1975completely} of the process $\map{C}$, defined as the state $C  :=  ( {\map C}   \otimes\map I)   (  \kketbbra I )/d_{\rm in}$, where $d_{\rm in}$ is the total dimension of all the input systems, and $\kket I =  \sum_{j=1}^{d_{\rm in}}  \ket j \otimes \ket j $ is the canonical (unnormalized) maximally entangled state. The rank of the Choi state $C$ is called the Kraus rank of the process $\map{C}$. Since a unitary evolution, namely a process without information loss, has Kraus rank equal to one, the Kraus rank could be interpreted as the degree of information loss introduced by the process. 

  An important parameter entering into the analysis is the degree of independence between systems, defined in the following.    
    For a state $\rho$, we say two disjoint subsystems $S$ and $T$ are independent if $\rho_{ST} = \rho_{S} \otimes \rho_{T}$, where $\rho_S$, $\rho_T$ and $\rho_{ST}$ are the marginal states of $\rho$ on the subsystems $S$, $T$ and the joint system $ST$, respectively. The degree of independence between $S$ and $T$ is then defined as 
    $\chione\rho S T := \| \rho_{ST} - \rho_S \otimes \rho_T \|_1$ \,,
   where $\|X\|_1 := \Tr\left[\sqrt{X^\dag X}\right]$ denotes the trace norm. 
Clearly, $\chione\rho S T = 0$ if and only if $\rho_{ST} = \rho_{S} \otimes \rho_{T}$. More generally, the trace distance is related to the probability to distinguish the states, and thus $\chione\rho S T$ measures the probability that an observer correctly decides whether the subsystems are independent or correlated.   In the following, we will apply this definition to the Choi state $C$ of the process $\map C$.  With this choice,   $\chione C S T$ satisfies the conditions for a quantum causality measure, as defined in Ref. \cite{jia2018quantifying}.


To facilitate the efficiency analysis of our algorithm, we first put the process $\map{C}$ in a standard form where all input and output wires have the same dimension $d_A$. 
This standard form does not limit the generality of the quantum process we investigate. For a process $\map{C}$ with input dimensions $d_{A_1},\dots,d_{A_n}$ and output dimensions $d_{B_1},\dots,d_{B_n}$, we can pick $d_A=\max\{d_{A_1},\dots,d_{A_n},d_{B_1},\dots,d_{B_n}\}$ and regard each input or output wire of $\map{C}$ as a subspace of a $d_A$-dimensional system, thus transforming $\map{C}$ into a process whose wires all have dimension $d_A$. 
In our analysis, we will use the Kraus rank of the process in the standard form to characterize the information loss. We assume that the information loss is bounded, which is given by the following assumption:


\begin{assumption} \label{ass:dAdM} \label{ass:lowrank}
The Kraus rank of $\map{C}$, after transforming it into the standard form, is bounded by a polynomial of $n$.
\end{assumption}


To ensure that the algorithm outputs the correct causal unravelling, we further require that the process satisfies a form of causal faithfulness, meaning that the strength of causal relations is either zero or above a threshold. 
Using $\chi_1$ as a quantitative measure of correlation, we adopt the following assumption:

\begin{assumption}\label{ass:ST_threshold}
   There exists a number $\chi_{\min}>0$ such that, for any two disjoint sets of wires $S$ and $T$ being tested for independence, either
   \begin{enumerate}
      \item $\chione C S T = 0$, or
      \item $\chione C S T \geq \chi_{\min}$,
   \end{enumerate}
   where $C$ is the Choi state of the quantum process $\map{C}$.
\end{assumption}

The threshold $\chi_{\min}$ determines the resolution of the independence tests. To guarantee the correctness of \autoref{func:recursive}, the independence tests must be precise enough to detect correlations above this threshold with high probability.
The efficiency and correctness of \autoref{func:recursive} are given in the following theorem, whose proof is in \supplementref{app:unitary2}: 
\begin{theo}\label{thm:unitary2}
Under Assumptions \ref{ass:dAdM} and \ref{ass:ST_threshold}, for any confidence parameter $\kappa_0>0$, \autoref{func:recursive} satisfies the following conditions: 

\begin{enumerate}
   \item With probability $1-\kappa_0$, the output of \autoref{func:recursive} is a correct causal unravelling for $\map{C}$.
   \item The number of queries to $\map{C}$ is in the order of
    \begin{align}
        T_{\rm sample}=O\left(n^3 d_A^2 r_{\map{C}}^2 \chi_{\min}^{-4} \log (n \kappa_0^{-1}) \right)
    \end{align}
      where $r_{\map{C}}$ is the Kraus rank of $\map{C}$ in the standard form with all wires having dimension $d_A$.
   \item The computational complexity is in the order of $O(T_{\rm sample}n\log d_A)$.
\end{enumerate}
\end{theo}

 \autoref{thm:unitary2} guarantees that, under appropriate assumptions,  the sample complexity of our algorithm is polynomial in the number of input and output systems of the process under consideration.  This feature is in stark contrast with the exponential complexity of full process tomography. As a consequence, our algorithm offers a speedup over  for other algorithms, such as the one proposed in  Ref. \cite{giarmatzi2018quantum}, which require  process tomography as an intermediate step. 

Assumptions \ref{ass:dAdM} and \ref{ass:ST_threshold} guarantee that \autoref{func:recursive} produces an exact causal unravelling. However, both assumptions can be lifted if we only require an approximate causal unravelling, meaning that the process $\map{C}$ is within a certain error of another process compatible with the causal unravelling output by the algorithm. In \supplementref{app:approximate}, we formulate this approximate case and prove that the error is small under the condition that every marginal Choi state of $\map{C}$ obtained by taking only the first $k$ inputs and $k-1$ outputs in the causal unravelling of $\map{C}$ has polynomial rank up to a small error. This condition can be verified efficiently during the execution of \autoref{func:recursive}.

\medskip 
{\bf Causal unravelling with local observations.} 
We now show that, under some assumptions on the input-output relations, one can design  algorithms that have much lower sample complexity in terms of $n$ compared to \autoref{func:recursive}, and are more experimentally friendly, in that they require only local state preparation and local measurements.

In the Methods, we show an efficient algorithm to detect the pairwise correlations between input and output wires of $\map{C}$ with local state preparation and local measurements. The algorithm computes a Boolean matrix $\ind_{ij}$ such that, with high probability, for every $i$ and $j$, $A_i$ and $B_j$ are approximately independent whenever $\ind_{ij}=\True$, and are correlated whenever $\ind_{ij}=\False$.
With some assumptions, this Boolean matrix $\ind_{ij}$ is sufficient to give the exact causal unravelling. The first case is given by the following assumption on the process $\map{C}$:

\begin{assumption} \label{ass:totalorder}
$\map{C}$ is a quantum comb in $\set{Comb}[(A_{\sigma(1)},B_{\pi(1)}),\dots,(A_{\sigma(n)},B_{\pi(n)})]$, and there exists a constant $\chi_{\min}>0$ such that, 
for any pair of input and output wires $A_{\sigma(i)}$ and $B_{\pi(j)}$, if $j\geq i$, then $\chione C{A_{\sigma(i)}}{B_{\pi(j)}} \geq \chi_{\min}$.
\end{assumption}

\autoref{ass:totalorder} indicates a non-trivial correlation between any pair consisting of an input system and an output system,  with the property that the input system appears before the output system in the overall causal order. In other words, for any $j \geq i$, $C_{A_{\sigma(i)},B_{\pi(j)}}$ is away from $C_{A_{\sigma(i)}} \otimes C_{B_{\pi(j)}}$ by distance  $\chi_{\min}$. Meanwhile, \autoref{ass:totalorder} defines a total order of the input (output) wires, and ensures a unique causal unravelling that $\map{C}$ is compatible with. 
Under this assumption, if $A_i$ is the $k$-th input, namely $i=\sigma(k)$, $A_i$ is correlated with $n-k+1$ output wires including every output $B_{\pi(j)}$ with $j\geq k$, and is independent of the other outputs. 
In other words, if we find $A_i$ is correlated with exactly $c_A(i)$ output wires, it must be the $(n-c_A(i)+1)$-th input. With this, the order of input wires can be exactly determined, and a similar statement can be applied to order the output wires.

In \supplementref{app:totalorder}, we give the details of this algorithm, and analyze its efficiency given by the following theorem:

\begin{theo}\label{thm:order_linear}
For a quantum comb $\map{C} \in \set{Comb}[(A_{\sigma(1)},B_{\pi(1)}),\dots,(A_{\sigma(n)},B_{\pi(n)})]$ satisfying \autoref{ass:totalorder}, there is an algorithm that satisfies the following conditions:
\begin{enumerate}
   \item With probability $1-\kappa$, the algorithm outputs the correct causal unravelling $(A_{\sigma(1)},B_{\pi(1)}),\dots,(A_{\sigma(n)},B_{\pi(n)})$.
   \item The algorithm uses only local state preparations and local measurements and the number of queries to $\map{C}$ is in the order of
\begin{align} \label{eq:order_linear}
    N = O\left(d_A^6 d_B^6 \chi_{\min}^{-2} \log(n d_A d_B \kappa^{-1}) \right) \,,
\end{align}
   where $d_A := \max_i d_{A_i}, d_B := \max_j d_{B_j}$. 
   \item The computational complexity is in the order of $O( Nn(n + d_A + d_B^4))$.
\end{enumerate}

\end{theo}

Note the sample complexity of this algorithm grows only logarithmically with $n$.

A similar idea could be adopted to the case where the process belongs to a special case of Markovian processes \cite{costa2016quantum,giarmatzi2018quantum,pollock2018operational,berk2021resource}, where each output depends only on one previous input and each input affects only one later output. This indicates that the process is decomposable to a tensor product of $n$ channels each with one input and one output. 
In our problem, the order of inputs and outputs is unknown, and we have the following assumption:

\begin{assumption} \label{ass:memoryless}
The process $C$ is a tensor product of $n$ channels, $\map{C} = \bigotimes_{i=1}^n \map{C}_i$ with $\map{C}_i: \spc{H}_{A_{i }}\to \spc{H}_{B_{\pi'(i)}}$ for some permutation $\pi'$.
\end{assumption}

Since each output is related to at most one input and each input affects at most one output, after obtaining $\ind_{ij}$, we can obtain the causal unravelling by matching each input-output pair $(A_i,B_j)$ with $\ind_{ij}=\False$. 

If there exists a threshold $\chi_{\min} >0$ such that either $\chione C{A_i}{B_j}=0$ or $\chione C{A_i}{B_j} \geq \chi_{\min}$ holds for every $A_i$ and $B_j$, then the algorithm has the same complexity as in \autoref{thm:order_linear} that is logarithmic in $n$. However, in case a threshold $\chi_{\min}$ is not known, we can still show that the algorithm is efficient yet produces an approximate answer with an error bound defined by the diamond norm  \cite{kitaev2002classical}. The diamond norm, also known as the completely bounded trace norm, is a distance measure between channels defined for $\map{C},\map{D}: L(\spc{H}_A) \to L(\spc{H}_B)$ as $\| \map{C} - \map{D} \|_\diamond := \max_{\rho \in S(\spc{H}_A \otimes \spc{H}_A)} \| (\map{C}\otimes \map{I}_A) (\rho) - (\map{D}\otimes \map{I}_A) (\rho) \|_1$, where $\map{I}_A:L(\spc{H}_A) \to L(\spc{H}_A)$ is the identity map. The diamond norm measures the maximum probability to distinguish two channels, and is tighter than the trace distance since $\|\map{C} - \map{D} \|_\diamond \leq \|C - D\|_1$ for all channels $\map{C}$ and $\map{D}$ with Choi states $C$ and $D$. 
The error bound is stated in the following theorem, whose proof is in \supplementref{app:memoryless}.

\begin{theo}\label{thm:memoryless}
For a quantum process $\map{C}$ satisfying \autoref{ass:memoryless}, there is an algorithm that outputs a causal unravelling $(A_1,B_{\pi(1)}),\dots,(A_n,B_{\pi(n)})$ satisfying the following conditions:
\begin{enumerate}
   \item With probability $1-\kappa$, the causal unravelling is approximately correct in the following sense:
      \begin{align}
          \exists \map{D} \in \set{Comb}[ (A_1,B_{\pi(1)}),\dots,(A_n,B_{\pi(n)}) ], ~ \|\map{C} - \map{D} \|_\diamond \leq \varepsilon \,.
      \end{align}
   \item The algorithm uses only local state preparations and local measurements and the number of queries to $\map{C}$ is in the order of
\begin{align}
    N = O\left(n^2 d_A^8 d_B^6 \varepsilon^{-2} \log(n d_A d_B \kappa^{-1}) \right)\,,
\end{align}
   where $d_A := \max_i d_{A_i}$ and $d_B := \max_j d_{B_j}$. 
   \item The computational complexity is in the order of $O( Nn(n + d_A + d_B^4))$.
\end{enumerate}
\end{theo}

\medskip 
{\bf Causal unravelling with interactions between more inputs and outputs.} 
In the algorithms shown so far, we assumed that the process under consideration admits a causal unravelling where each interaction involves exactly one input and one output. More generally, \autoref{func:recursive} can be easily extended to the scenario where each interaction involves at most $c$ inputs and $c$ outputs of the original process. 
Instead of considering each wire separately, the idea is to consider a subset of at most $c$ input (output) wires and perform independence tests on the subsets.

In \autoref{func:recursive}, one enumerates an input-output pair $(A_x,B_y)$ and checks whether it is the last tooth by performing an independence test between $A_x$ and $A_{\neq x}B_{\neq y}$. To deal with larger $c$, we replace this procedure by enumerating a subset of input wires $P\subset\{A_1,\dots,A_{n_{\rm in}}\}$ and a subset of output wires $Q\subset\{B_1,\dots,B_{n_{\rm out}}\}$, satisfying $|P|\leq c$ and $|Q|\leq c$. Then we check whether $(P,Q)$ is the last tooth of $\map{C}$, which, according to \autoref{prop:last_tooth_informal}, is equivalent to checking the independence between $P$ and $A_{\notin P}B_{\notin Q}:=\{A_i|A_i\notin P\}\cup\{B_j|B_j \notin Q\}$. The independence tests can still be implemented with SWAP tests.
Like \autoref{func:recursive}, after we decide $(P,Q)$ to be the last tooth, the wires in $P$ and $Q$ are removed from consideration, and the problem is reduced to the causal unravelling of a smaller channel. This process is done recursively until one reaches the bottom case. 
We give the detailed algorithm and analysis in \supplementref{app:largerc}. For constant $c$, under some assumptions, the complexity of the algorithm is still polynomial in $n$ and $d_A$, while the exponent depends on $c$.

\section{Discussion}


In this paper we developed an efficient algorithm for  discovering  linear causal structures between the inputs and outputs of a multipartite quantum process.  Our algorithm provides a partial solution to the more general quantum causal discovery problem, whose goal is to produce a full causal graph describing arbitrary causal correlations in an arbitrary  set of quantum variables.  Our algorithm can be used as the first step for quantum causal discovery, and to obtain the full causal structure, additional tests may be adopted to detect signalling between more subsets of input and output wires. 
Since the most general quantum causal discovery problem is intrinsically hard, an interesting direction for future work is  to examine  to what extent the problem of quantum causal discovery can be solved by an efficient algorithm in scenarios beyond the linear structure analyzed in this work.

The efficiency of our algorithms relies on some assumptions. For \autoref{func:recursive}, the low-rank assumption is the key in both the exact case (\autoref{ass:lowrank}) and the approximate case discussed in \supplementref{app:approximate}. This assumption avoids the computational difficulty of deciding whether a completely general state is a product state \cite{gutoski2013quantum}. Physically, a quantum process has a low rank if the number of uncontrolled particles entering and/or exiting the interaction region is small. In this picture, the uncontrolled  particles in the input can be  regarded as sources of environmental noise, and the uncontrolled  particles in the output are responsible for information loss in the process. 
Intuitively, without the low-rank assumption, the causal correlations will be obscured by  the noise, and it will be hard to discover them without additional prior knowledge.
On the other hand, if one has prior knowledge, as in the case of processes satisfying Assumptions \ref{ass:totalorder} and \ref{ass:memoryless}, the low-rank assumption may be lifted.




The ability to infer the underlying causal structure of a process is useful for a variety of applications. 
Classically, discovering causal relationships is the goal of many research areas with numerous applications in social and biomedical sciences \cite{spirtes2000causation} such as the construction of the gene expression network \cite{spirtes2000constructing}. Causal discovery allows us to understand complex systems whose internal structures are not directly accessible, and to discover possible models for the internal mechanisms. Quantum causal discovery, likewise, enables the modelling of quantum physical processes with inaccessible internal structure, for example, discovering the individual interactions in scattering experiments. Below we list some specific examples where our causal unravelling algorithms can be applied to the detection of correlations and the modelling of internal structures of complex processes.




First, causal relations among quantum variables  are relevant to the study of  quantum networks \cite{kimble2008quantum,elliott2002building,wehner2018quantum}, where the presence of a causal relation between  two systems can be used to test whether it is possible to send signals  from one node to another.   
In a realistic setting, the signalling patterns within a quantum network may change dynamically, depending on the number of users of the network at a given moment of time, on the way the messages are routed from the senders to the receivers, and also on changes in the environment, which may affect the  availability of transmission paths between nodes.   Such a dynamical structure occurs frequently in classical wireless networks \cite{johnson1996dynamic,royer1999review}, and is likely to arise in a future quantum internet. In this context, our algorithms provide   an efficient way to detect dynamical changes in  the availability of data transmission paths.


The detection of causal relations is also relevant to the verification of quantum devices,  as it can be used as an initial  test to determine  whether a given quantum device  generates input-output correlations with a desired causal structure. Such a test could serve as an initial screening to rule out devices that are not suitable for a given task, and could be followed by more refined  quantum benchmarks  \cite{bai2018test} which quantify how well the device performs a desired task.   In this context, the benefit of the causal unravelling test  is that it could save the effort of performing more refined tests in case the process under consideration does not comply with the desired causal structure.


Finally, our causal unravelling algorithm can be used as a preliminary step to full process tomography.  
 By detecting the causal structure of multipartite quantum processes, one can sometimes design a tailor-made tomography scheme that ignores   unnecessary correlations, and achieves full process tomography without requiring an exponentially large number of measurement setups.  For example, a process that admits a causal unravelling with systems of bounded dimension at every step    can be efficiently represented by a tensor network state \cite{fannes1992finitely,verstraete2008matrix}, for which   tomography  can be performed efficiently \cite{cramer2010efficient}.  
In a quantum communication network, efficient tomography of the transmission paths is crucial for the design of encoding, decoding and calibration schemes for more efficient data transmission. In physics experiments, the causal structure and tomography data are useful for modelling the underlying physical process, for example, by finding the smallest quantum model that reproduces the observed data \cite{gu2012quantum,monras2016quantum,thompson2017using}.
More generally, characterizing the causal structure of a multipartite process as a tensor network enables the use of efficient protocols that exploits the tensor network structure, including simulation protocols \cite{verstraete2008matrix,shi2006classical,vidal2008class} and compression protocols \cite{bai2020quantum}.

\section{Methods}

{\bf Efficient tests for the last tooth via the SWAP test.} 
In the Results, we have given the framework of \autoref{func:recursive}. In this section, we discuss 
 how the tests for the last tooth, namely \autoref{line:ind} of \autoref{func:recursive}, can be carried out efficiently.

Consider a process $\map{C}$ of three input wires $A_1,A_2,A_3$ and three output wires $B_1,B_2,B_3$, and suppose that we want to test  whether $(A_1,B_1)$ is the last tooth, which is equivalent to the independence test between $A_1$ and $A_2A_3B_2B_3$ according to \autoref{prop:last_tooth_informal}. Testing the independence between $A_1$ and $A_2A_3B_2B_3$ can be converted to the estimation of $\chione{C}{A_1}{A_2A_3B_2B_3} = \|C_{A_1A_2A_3B_2B_3} - C_{A_1}\otimes C_{A_2A_3B_2B_3}\|_1$, which is the distance between marginal Choi states. Each copy of $C_{A_1A_2A_3B_2B_3}$ or $C_{A_2A_3B_2B_3}$ can be prepared with one use of the process $\map{C}$. Note that $C_{A_1}$ equals to $I_{A_1}/d_{A_1}$ by definition of a CPTP map. Given the ability to prepare the marginal Choi states, we now consider the estimation of their distance, which gives the value of $\chione{C}{A_1}{A_2A_3B_2B_3}$. In the following, we first talk about the estimation of another distance measure, the Hilbert-Schmidt distance, and use it to bound the trace distance as used by $\chi_1$.

\tikzset{C1to4/.pic={
    \node[tensor, minimum height=20mm] (C2) at (0,0) {$\map{C}$};
    \coordinate (C1) at ($(C2)+(0,0.75)$);
    \coordinate (C0) at ($(C1)+(0,0.75)$);
    \coordinate (C3) at ($(C2)-(0,0.75)$);
    \coordinate (C4) at ($(C3)-(0,0.75)$);
    \coordinate (C5) at ($(C4)-(0,0.75)$);
    \foreach \i in {0,1,2,3,4,5}
    {
        \coordinate (A\i) at ($(C\i-|C2.west)+(-0.85,0)$);
        \coordinate (B\i) at ($(C\i-|C2.east)+(0.85,0)$);
    }
    \coordinate (grd) at ($(B1)+(0.25,-0.1)$);
    \foreach \i in {1,2,3}
    {
        \draw[thick] (C\i-|C2.west) -- node[above, midway]{$A_{\i}$} (A\i);
        \draw[thick] (C\i-|C2.east) -- node[above, midway]{$B_{\i}$} (B\i);
    }
    \pic at (grd) {ground};
    \draw[thick] (grd) |- (B1);
    \node[prepare tall, minimum height=25mm, anchor=east] (phi) at ($(A3)!0.5!(A4)-(0,0)$) {$\ket{\Phi_+}_{A_2A_3A_2'A_3'}$};
    \draw[thick] (phi.east|-A2) -- (A2);
    \draw[thick] (phi.east|-A3) -- (A3);
    \draw[thick] (phi.east|-A4) -- (C2.east|-A4) -- node[above, midway]{$A_2'$} (B4);
    \draw[thick] (phi.east|-A5) -- (C2.east|-A5) -- node[above, midway]{$A_3'$} (B5);
    }
}

Generally, the Hilbert-Schmidt distance between two states $\rho$ and $\sigma$ is defined as $\|\rho - \sigma\|_2$, where $\|X\|_2:= \sqrt{\Tr[X^\dag X]}$ denotes the Schatten 2-norm, also known as the Frobenius norm. It is related to the trace distance by the following inequality \cite{coles2019strong}:
\begin{align}\label{eq:norm12_main}
    2\|\rho - \sigma\|_2^2 \leq \|\rho - \sigma\|_1^2 \leq \frac{4\rank(\rho)\rank(\sigma)}{\rank(\rho)+\rank(\sigma)} \|\rho - \sigma\|_2^2 \,.
\end{align}

The Hilbert-Schmidt distance between two quantum states can be estimated via SWAP tests \cite{buhrman2001quantum}. 
  The SWAP test uses the quantum circuit in \autoref{fig:SWAP} to estimate $\Tr[\rho\sigma]$ for two given quantum states $\rho$ and $\sigma$.

\newcommand{\cS}{\mathsf{CSWAP}}
\newcommand{\swap}{\mathsf{SWAP}}

\begin{figure}[h]
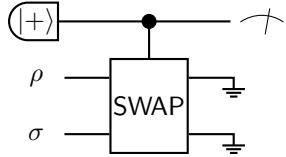

\iffigures
$$\begin{array}{ccc}
    \begin{mathtikz}
        \node[tensor, minimum height=12.5mm] (swap) at (0,-0.375) {$\swap$};
        \node (phi1) at (-1.5,0) {$\rho$};
        \node (phi2) at (-1.5,-0.75) {$\sigma$};
        \node[prepare] (c) at (-1.5,0.75) {$\ket{+}$};
        \node (right) at (1.125, 0) {};
        \node (left) at (-1.125, 0) {};
        \node (measure) at (1.5,0.75) {\tikz{\measurement}};
        \draw[thick] (phi1-|left) -- (phi1-|swap.west);
        \draw[thick] (phi2-|left) -- (phi2-|swap.west);
        \coordinate (grd1) at ($(right|-phi1) - (0,0.15)$);
        \coordinate (grd2) at ($(right|-phi2) - (0,0.15)$);
        \pic at (grd1) {ground};
        \pic at (grd2) {ground};
        \draw[thick] (swap.east|-phi1) -| (grd1);
        \draw[thick] (swap.east|-phi2) -| (grd2);
        \draw[thick] (c) -- (measure);
        \fill (swap|-c) circle [radius=0.1];
        \draw[thick] (swap|-c) -- (swap);
    \end{mathtikz}
\end{array}$$
\fi
\caption{\label{fig:SWAP}{\bf The SWAP test circuit.} The circuit consists of a controlled-SWAP gate with control qubit initialized to $\ket+$. Measuring the control system under the $\{\ket{+},\ket{-}\}$ basis yields the outcome $\ket{+}$ with probability $(1+\Tr[\rho\sigma])/2$. The ground symbol means discarding the system.}
\end{figure}

  If we run the circuit in \autoref{fig:SWAP} for $N$ times and let $c_+$ be the number of times observing outcome $\ket +$, then $2c_+/N - 1$ is an estimate of $\Tr[\rho\sigma]$. 
  The algorithm that yields an estimate of $\Tr[\rho\sigma]$ is as follows, which produces an estimate with error no more than $\varepsilon$ with probability $1-\kappa$ as shown in \autoref{lem:SWAP}.

  \begin{function}[H]
      \caption{SWAPTEST($\rho$, $\sigma$, $\varepsilon$, $\kappa$)}\label{func:SWAPTEST}
      \Indentp{0.5em}
      \SetInd{0.5em}{1em}
      \SetKwInOut{Preprocessing}{Preprocessing}
      \SetKwInOut{Input}{Input}
      \SetKwInOut{Output}{Output}
      \SetKw{Next}{next}
      \SetKwData{ind}{ind}
      \SetKwData{last}{last}
      \SetKwData{purity}{purity}
      \SetKwData{maxpurity}{maxpurity}

      \ResetInOut{Input}
      \Input{Quantum states $\rho$ and $\sigma$ (accessed by oracles that generate the states), error threshold $\varepsilon$, confidence $\kappa$}
      \ResetInOut{Output}
      \Output{Approximate value of $\Tr[\rho\sigma]$}
      \BlankLine
      $N \gets \lceil 2\varepsilon^{-2}\log(2/\kappa) \rceil$\;
      Run the circuit in \autoref{fig:SWAP} for $N$ times. Let $c_+$ be the number of outcome $\ket +$\;
      Return $2c_+/N - 1$\;
  \end{function}

  \begin{lem}\label{lem:SWAP}
  With probability $1-\kappa$, the SWAP test estimates $\Tr[\rho\sigma]$ within error $\varepsilon$.
  \end{lem}
  \begin{proof}
  For each run, the probability that the outcome is $\ket +$ is $(1+\Tr[\rho\sigma])/2$.
  By Hoeffding's inequality,
  \begin{align}
  \nonumber    \Pr \left[ \left| \frac{c_+}{N} - \frac{1+\Tr[\rho\sigma]}{2} \right| \leq \varepsilon/2 \right] & \geq 1-2e^{-\varepsilon^2 N /2} \\
      \Pr \left[ \left| 2c_+/N - 1 - \Tr[\rho\sigma] \right| \leq \varepsilon \right] & \geq 1-2e^{-\varepsilon^2 ( 2\varepsilon^{-2}\log(2/\kappa)) /2} = 1-\kappa
  \end{align}
  \end{proof}

  The SWAP test is efficient in the sense that its circuit complexity is linear in the number of qubits representing the systems. 
  For $d$-dimensional states $\rho$ and $\sigma$, the controlled-SWAP gate acting on $\rho$ and $\sigma$ can be realized with $\lceil \log d \rceil$ controlled-SWAP gates acting on each pair of corresponding qubits of $\rho$ and $\sigma$, and thus can be implemented with $O(\log d)$ gates. 

  Running SWAP tests on $\rho \otimes \sigma$, $\rho \otimes \rho$, and $\sigma\otimes\sigma$, we are able to obtain estimates for $\Tr[\rho\sigma]$, $\Tr[\rho^2]$ and $\Tr[\sigma^2]$, respectively. From these estimates we can compute the Hilbert-Schmidt distance 
  between $\rho$ and $\sigma$ 
  according to the following equation:
  \begin{align}\label{eq:HS_from_Tr}
      \|\rho - \sigma\|_2^2 = \Tr[\rho^2] + \Tr[\sigma^2] - 2\Tr[\rho\sigma]\,.
  \end{align}

Coming back to the independence tests, by applying SWAP tests on marginal Choi states $C_{A_1A_2A_3B_2B_3}$ and $C_{A_1}\otimes C_{A_2A_3B_2B_3}$, we can estimate their Hilbert-Schimidt distance. Using Eq. (\ref{eq:norm12_main}), we can bound their trace distance, obtain a bound on $\chione{C}{A_1}{A_2A_3B_2B_3}$, determine the independence between $A_1$ and $A_2A_3B_2B_3$, and decide whether $(A_1,B_1)$ is the last tooth. This procedure is summarized in \autoref{func:findlast_direct}.

\begin{function}[H]
      \caption{checklast($\map{C}$, $A_x$, $B_y$, $\varepsilon$, $\delta$, $\kappa$)}\label{func:findlast_direct}
      \Indentp{0.5em}
      \SetInd{0.5em}{1em}
      \SetKwInOut{Preprocessing}{Preprocessing}
      \SetKwInOut{Input}{Input}
      \SetKwInOut{Output}{Output}
      \SetKw{Next}{next}
      \SetKw{Break}{break}
      \SetKwData{ind}{ind}
      \SetKwData{last}{last}
      \SetKwData{purity}{purity}
      \SetKwData{maxpurity}{maxpurity}
      \SetKwFunction{SWAPTEST}{SWAPTEST}

      \ResetInOut{Input}
      \Input{A quantum channel $\map{C}$, input wire $A_x$, output wire $B_y$, error thresholds $\varepsilon,\delta$, confidence $\kappa$}
      \ResetInOut{Output}
      \Output{Whether $(A_x,B_y)$ is the last tooth of $\map{C}$}
      \BlankLine

      Prepare the circuits that generate the states $C_1 := \Tr_{B_y}[C]$ and $C_2 := I_{A_x}/d_{A_x} \otimes \Tr_{A_xB_y}[C]$\;

      $p_{1} \gets $\SWAPTEST{$C_1$, $C_1$, $\varepsilon$, $\kappa$} \;
      $p_{2} \gets $\SWAPTEST{$C_2$, $C_2$, $\varepsilon$, $\kappa$} \;
      $p_{3} \gets $\SWAPTEST{$C_1$, $C_2$, $\varepsilon$, $\kappa$} \;
      \uIf(\tcp*[f]{The estimated distance is larger than $\delta$}){$p_{1} + p_{2} - 2 p_{3} > \delta$}
      {
          Return \False \tcp*{$(A_x,B_y)$ is not the last tooth}
      }
      \Else
      {
          Return \True \tcp*{$(A_x,B_y)$ is the last tooth}
      }
  \end{function}

This completes \autoref{func:recursive}. The choices of $\varepsilon$ and $\delta$ are done in the detailed error analysis in \supplementref{app:unitary2}.
To ensure the correctness and efficiency, we need to further ensure that Eq. (\ref{eq:norm12_main}) provides a useful bound by giving upper bounds on the ranks of the marginal Choi states. Such upper bounds can be obtained from \autoref{ass:lowrank}, and the details are in \supplementref{app:unitary2}.

We have shown the possibility of using the SWAP test as an efficient test of independence.
The SWAP test could in principle be replaced by any algorithm for estimating the Hilbert-Schmidt distance, with possibly lower sample complexity \cite{buadescu2019quantum}.
In a quantum network scenario, causal unravelling may be performed by multiple parties, each of which has access to one input system or one output system. 
A bonus of the SWAP test is that, since a controlled-SWAP gate for two multipartite states can be decomposed into controlled-SWAP gates on the local systems of each party, only the control qubit needs to be transferred from one party to another in each round of the test. Additionally, the control qubit is measured after each round, and therefore parties do not need to carry quantum memories over subsequent rounds.

\medskip 
{\bf Testing the independence between all input-output pairs.} 
In this section, we show the algorithm to test the independence between every pair of input and output wires, resulting in a Boolean matrix $\ind_{ij}$. 
The first step is to  collect statistics of the channel $\map{C}$.  When collecting data, we need to ensure that the information encoded in the quantum data are preserved, for which purpose we use informationally complete POVMs \cite{prugovevcki1977information}. The elements of such a POVM on Hilbert space $\spc{H}$ form a spanning set of $L(\spc{H})$, and the number of elements can be chosen as $(\dim\spc{H})^2$. 

We pick an informationally complete POVM $\{P_{A_i,\alpha}\}_{\alpha=1}^{d_{A_i}^2}$ for each input wire $A_i$ and $\{Q_{B_j,\beta}\}_{\beta=1}^{d_{B_j}^2}$ for each output wire $B_j$. We pick $N$ to be the number of queries to the channel $\map{C}$. In each query, we measure the Choi state of $\map{C}$ with the POVM
\begin{align}\label{eq:tomo_POVM}
  \left\{P_{A_1,\alpha_1} \otimes \dots \otimes P_{A_n,\alpha_n} \otimes Q_{B_1,\beta_1} \otimes \dots \otimes Q_{B_n,\beta_n}\right\}_{\alpha_1,\dots,\alpha_n,\beta_1,\dots,\beta_n}
\end{align}
This POVM can be realized with local measurements on each wire. Let $(\alpha_1^{(k)},\dots,\alpha_n^{(k)},\beta_1^{(k)},\dots,\beta_n^{(k)})$ be the outcome of the $k$-th query. We can write the outcomes in a matrix as shown in \autoref{fig:matrix}.

\begin{figure}[h]
\centering
\iffigures
\begin{tikzpicture}
    \matrix[row sep={7.5mm,between origins},nodes={inner sep=3mm}]
    {
        \node(r1){$k$}; & \node(c2){$A_1$}; & \node(c3){$\cdots$}; & \node(c4){$A_i$}; & \node(c5){$\cdots$}; & \node(c6){$A_n$}; & \node(c7){$B_1$}; & \node(c8){$\cdots$}; & \node(c9){$B_j$}; & \node(c10){$\cdots$}; & \node(c11){$B_n$};\\
        \node(r2){1}; & \node{$\alpha_1^{(1)}$}; & \node{$\cdots$}; & \node{$\alpha_i^{(1)}$}; & \node{$\cdots$}; & \node{$\alpha_n^{(1)}$}; & \node{$\beta_1^{(1)}$}; & \node{$\cdots$}; & \node{$\beta_j^{(1)}$}; & \node{$\cdots$}; & \node{$\beta_n^{(1)}$};\\
        \node(r3){2}; & \node{$\alpha_1^{(2)}$}; & \node{$\cdots$}; & \node{$\alpha_i^{(2)}$}; & \node{$\cdots$}; & \node{$\alpha_n^{(2)}$}; & \node{$\beta_1^{(2)}$}; & \node{$\cdots$}; & \node{$\beta_j^{(2)}$}; & \node{$\cdots$}; & \node{$\beta_n^{(2)}$};\\
        \node(r4){$\vdots$}; & \node{$\vdots$}; & \node{$\ddots$}; & \node{$\vdots$}; & \node{$\ddots$}; & \node{$\vdots$}; & \node{$\vdots$}; & \node{$\ddots$}; & \node{$\vdots$}; & \node{$\ddots$}; & \node{$\vdots$};\\
        \node(r5){$N$}; & \node{$\alpha_1^{(N)}$}; & \node{$\cdots$}; & \node{$\alpha_i^{(N)}$}; & \node{$\cdots$}; & \node{$\alpha_n^{(N)}$}; & \node{$\beta_1^{(N)}$}; & \node{$\cdots$}; & \node{$\beta_j^{(N)}$}; & \node{$\cdots$}; & \node{$\beta_n^{(N)}$};\\
    };
    \foreach \c in {2,3,5,6,7,8,10,11}
    {
        \fill[white,opacity=0.67] (c\c.west|-r1.north) rectangle (c\c.east|-r5.south);
    }
    \draw (r1.south west) -- (c11.south east);
    \draw (r1.north east) -- (r5.south east);
    \draw[dashed] (c4.north west) rectangle (c4.east|-r5.south);
    \draw[dashed] (c9.north west) rectangle (c9.east|-r5.south);
    \draw[<->] (c4.center|-r5.south) -- +(0,-0.5) -- node[above]{Independent?} ($(c9.center|-r5.south)+(0,-0.5)$) -- (c9.center|-r5.south);
\end{tikzpicture}
\fi
\caption{\label{fig:matrix}{\bf Matrix of measurement outcomes.} Each row corresponds to a query to the process, and each column corresponds to an input or output wire. When we investigate $A_i$ and $B_j$, we only look at the corresponding columns in the matrix and check the independence from data in these columns.}
\end{figure}

This matrix will be used to determine the independence between every pair of input and output wires. When we investigate $A_i$ and $B_j$, we only look at the corresponding columns in the matrix and check the independence from data in these columns. Since we are using an informationally complete POVM, the independence between the two columns of the matrix is equivalent to the independence between the input $A_i$ and output $B_j$.
With this idea, we will be able to compute a Boolean matrix $\ind_{ij}$ such that $\ind_{ij}=\True$ if and only if $A_i$ and $B_j$ are independent, up to an error related to $N$, the number of queries.

One may notice that the POVM in Eq. (\ref{eq:tomo_POVM}) is enough for a full process tomography of $\map{C}$. This is true if we pick an exponentially large $N$ as large as the square of the product of all input and output dimensions.
However, to compute $\ind_{ij}$ with a small error, $N$ does not need to be large. We show that $N$ could be chosen to grow only logarithmically with $n$, far less than a full process tomography.

\begin{lem}\label{lem:ind_informal}
Given a channel $\map{C}$ with input wires $A_1,\dots,A_n$ and output wires $B_1,\dots,B_n$, there exists an algorithm that satisfies the following conditions:
\begin{enumerate}
   \item With probability $1-n^2\kappa_0$, the algorithm produces an output satisfying:
      \begin{enumerate}
          \item if $\ind_{i,j} = \False$, then $\chi_1(C_{A_i,B_j}) > \chi_- -\varepsilon_0$
          \item if $\ind_{i,j} = \True$, then $\chi_1(C_{A_i,B_j}) \leq \chi_- + \varepsilon_0$
      \end{enumerate}
      where $\chi_->0$ is a threshold that can be freely chosen.
   \item The number of queries to $C$ is in the order of
      \begin{align}
         N = O\left(d_A^6 d_B^6\varepsilon_0^{-2} \log(d_A d_B \kappa_0^{-1}) \right)\,.
      \end{align}
\end{enumerate}
\end{lem}

The detailed algorithm and the proof are in \supplementref{app:all_pairs}.


The Boolean matrix $\ind_{ij}$ gives a partial order of the wires: if $\ind_{ij}=\False$, then $A_i$ must be before $B_j$. This partial order may not produce the full ordering of the wires, but will be a convenient initial guess for the causal unravelling, since its complexity is much lower than the general algorithm \autoref{func:recursive} for large $n$. In the Results, we show that under Assumptions \ref{ass:totalorder} or \ref{ass:memoryless}, this partial order is enough to infer the full order.


Furthermore, the matrix of outcomes (\autoref{fig:matrix}) is a conversion from the quantum process to classical data, making it possible to adopt classical causal discovery algorithms \cite{heinze2018causal,spirtes2000causation,chickering2002optimal,tsamardinos2006max}. For example, one could use algorithms based on the graph model, and try to find a graph in the form of \autoref{fig:bipartite_graph} that best fits the observation data in \autoref{fig:matrix}. 


\section{Data availability}
The authors declare that the data supporting the findings of this study are available within the paper and in the Supplementary Information files.


\section{Acknowledgements} 
This work was supported by the National Natural Science Foundation of China through grant 11675136, the Hong Kong Research Grant Council through grants 17300918 and 17307520, and though the Senior Research Fellowship Scheme SRFS2021-7S02, the Croucher Foundation, and the John Templeton Foundation through grant 61466, The Quantum Information Structure of Spacetime (qiss.fr). Research at the Perimeter Institute is supported by the Government of Canada through the Department of Innovation, Science and Economic Development Canada and by the Province of Ontario through the Ministry of Research, Innovation and Science. The opinions expressed in this publication are those of the authors and do not necessarily reflect the views of the John Templeton Foundation.
The work of M. Hayashi was supported in part by Guangdong Provincial Key Laboratory (Grant No. 2019B121203002).


\section{Author contributions}
G.B., Y.-D.W., Y.Z. and G.C. contributed to the development of the first algorithm. M.H. and G.B. proposed and refined the algorithms with local observations. G.C. proposed the problem and supervised this project. All authors discussed extensively the research presented in this paper and contributed to the writing of this manuscript.

\section{Competing interests}
The Authors declare no Competing Financial or Non-Financial Interests.

\bibliography{causal}

\begin{thebibliography}{10}

\bibitem{spirtes2000causation}
Peter Spirtes, Clark~N Glymour, Richard Scheines, and David Heckerman.
\newblock {\em Causation, prediction, and search}.
\newblock MIT press, 2000.

\bibitem{pearl2009causality}
Judea Pearl.
\newblock {\em Causality}.
\newblock Cambridge university press, 2009.

\bibitem{heinze2018causal}
Christina Heinze-Deml, Marloes~H Maathuis, and Nicolai Meinshausen.
\newblock Causal structure learning.
\newblock {\em Annual Review of Statistics and Its Application}, 5:371--391,
  2018.

\bibitem{chickering2002optimal}
David~Maxwell Chickering.
\newblock Optimal structure identification with greedy search.
\newblock {\em Journal of machine learning research}, 3:507--554, 2002.

\bibitem{tsamardinos2006max}
Ioannis Tsamardinos, Laura~E Brown, and Constantin~F Aliferis.
\newblock The max-min hill-climbing bayesian network structure learning
  algorithm.
\newblock {\em Machine learning}, 65:31--78, 2006.

\bibitem{chickering1996learning}
David~Maxwell Chickering.
\newblock Learning bayesian networks is np-complete.
\newblock In {\em Learning from data}, pages 121--130. Springer, 1996.

\bibitem{chickering2004large}
Max Chickering, David Heckerman, and Chris Meek.
\newblock Large-sample learning of bayesian networks is np-hard.
\newblock {\em Journal of Machine Learning Research}, 5:1287--1330, 2004.

\bibitem{wood2015lesson}
Christopher~J Wood and Robert~W Spekkens.
\newblock The lesson of causal discovery algorithms for quantum correlations:
  Causal explanations of bell-inequality violations require fine-tuning.
\newblock {\em New Journal of Physics}, 17(3):033002, 2015.

\bibitem{van2019quantum}
Thomas Van~Himbeeck, Jonatan~Bohr Brask, Stefano Pironio, Ravishankar
  Ramanathan, Ana~Bel{\'e}n Sainz, and Elie Wolfe.
\newblock Quantum violations in the instrumental scenario and their relations
  to the bell scenario.
\newblock {\em Quantum}, 3:186, 2019.

\bibitem{henson2014theory}
Joe Henson, Raymond Lal, and Matthew~F Pusey.
\newblock Theory-independent limits on correlations from generalized bayesian
  networks.
\newblock {\em New Journal of Physics}, 16(11):113043, 2014.

\bibitem{pienaar2015graph}
Jacques Pienaar and {\v{C}}aslav Brukner.
\newblock A graph-separation theorem for quantum causal models.
\newblock {\em New Journal of Physics}, 17(7):073020, 2015.

\bibitem{costa2016quantum}
Fabio Costa and Sally Shrapnel.
\newblock Quantum causal modelling.
\newblock {\em New Journal of Physics}, 18(6):063032, 2016.

\bibitem{allen2017quantum}
John-Mark~A Allen, Jonathan Barrett, Dominic~C Horsman, Ciar{\'a}n~M Lee, and
  Robert~W Spekkens.
\newblock Quantum common causes and quantum causal models.
\newblock {\em Physical Review X}, 7(3):031021, 2017.

\bibitem{barrett2019quantum}
Jonathan Barrett, Robin Lorenz, and Ognyan Oreshkov.
\newblock Quantum causal models.
\newblock {\em arXiv preprint arXiv:1906.10726}, 2019.

\bibitem{barrett2021cyclic}
Jonathan Barrett, Robin Lorenz, and Ognyan Oreshkov.
\newblock Cyclic quantum causal models.
\newblock {\em Nat. Commun.}, 12:885, 2021.

\bibitem{ried2015quantum}
Katja Ried, Megan Agnew, Lydia Vermeyden, Dominik Janzing, Robert~W Spekkens,
  and Kevin~J Resch.
\newblock A quantum advantage for inferring causal structure.
\newblock {\em Nature Physics}, 11:414--420, 2015.

\bibitem{fitzsimons2015quantum}
Joseph~F Fitzsimons, Jonathan~A Jones, and Vlatko Vedral.
\newblock Quantum correlations which imply causation.
\newblock {\em Scientific reports}, 5:1--7, 2015.

\bibitem{chiribella2019quantum}
Giulio Chiribella and Daniel Ebler.
\newblock Quantum speedup in the identification of cause--effect relations.
\newblock {\em Nature communications}, 10:1472, 2019.

\bibitem{giarmatzi2018quantum}
Christina Giarmatzi and Fabio Costa.
\newblock A quantum causal discovery algorithm.
\newblock {\em NPJ Quantum Information}, 4:1--9, 2018.

\bibitem{pollock2018operational}
Felix~A Pollock, C{\'e}sar Rodr{\'\i}guez-Rosario, Thomas Frauenheim, Mauro
  Paternostro, and Kavan Modi.
\newblock Operational markov condition for quantum processes.
\newblock {\em Physical review letters}, 120(4):040405, 2018.

\bibitem{berk2021resource}
Graeme~D Berk, Andrew~JP Garner, Benjamin Yadin, Kavan Modi, and Felix~A
  Pollock.
\newblock Resource theories of multi-time processes: A window into quantum
  non-markovianity.
\newblock {\em Quantum}, 5:435, 2021.

\bibitem{chiribella2008quantum}
Giulio Chiribella, G~Mauro D’Ariano, and Paolo Perinotti.
\newblock Quantum circuit architecture.
\newblock {\em Physical Review Letters}, 101(6):060401, 2008.

\bibitem{chiribella2009theoretical}
Giulio Chiribella, Giacomo~Mauro D’Ariano, and Paolo Perinotti.
\newblock Theoretical framework for quantum networks.
\newblock {\em Physical Review A}, 80(2):022339, 2009.

\bibitem{gutoski2013quantum}
Gus Gutoski, Patrick Hayden, Kevin Milner, and Mark~M Wilde.
\newblock Quantum interactive proofs and the complexity of separability
  testing.
\newblock {\em Theory of Computing}, 11(3):59, 2015.

\bibitem{buhrman2001quantum}
Harry Buhrman, Richard Cleve, John Watrous, and Ronald De~Wolf.
\newblock Quantum fingerprinting.
\newblock {\em Physical Review Letters}, 87(16):167902, 2001.

\bibitem{choi1975completely}
Man-Duen Choi.
\newblock Completely positive linear maps on complex matrices.
\newblock {\em Linear algebra and its applications}, 10(3):285--290, 1975.

\bibitem{jia2018quantifying}
Ding Jia.
\newblock Quantifying causality in quantum and general models.
\newblock {\em arXiv preprint arXiv:1801.06293}, 2018.

\bibitem{kitaev2002classical}
Alexei~Yu Kitaev, Alexander Shen, Mikhail~N Vyalyi, and Mikhail~N Vyalyi.
\newblock {\em Classical and quantum computation}.
\newblock Number~47. American Mathematical Soc., 2002.

\bibitem{spirtes2000constructing}
Pater Spirtes, Clark Glymour, Richard Scheines, Stuart Kauffman, Valerio
  Aimale, and Frank Wimberly.
\newblock Constructing bayesian network models of gene expression networks from
  microarray data.
\newblock 2000.

\bibitem{kimble2008quantum}
H~Jeff Kimble.
\newblock The quantum internet.
\newblock {\em Nature}, 453(7198):1023--1030, 2008.

\bibitem{elliott2002building}
Chip Elliott.
\newblock Building the quantum network.
\newblock {\em New Journal of Physics}, 4(1):46, 2002.

\bibitem{wehner2018quantum}
Stephanie Wehner, David Elkouss, and Ronald Hanson.
\newblock Quantum internet: A vision for the road ahead.
\newblock {\em Science}, 362(6412), 2018.

\bibitem{johnson1996dynamic}
David~B Johnson and David~A Maltz.
\newblock Dynamic source routing in ad hoc wireless networks.
\newblock In {\em Mobile computing}, pages 153--181. Springer, 1996.

\bibitem{royer1999review}
Elizabeth~M Royer and Chai-Keong Toh.
\newblock A review of current routing protocols for ad hoc mobile wireless
  networks.
\newblock {\em IEEE personal communications}, 6(2):46--55, 1999.

\bibitem{bai2018test}
Ge~Bai and Giulio Chiribella.
\newblock Test one to test many: a unified approach to quantum benchmarks.
\newblock {\em Physical Review Letters}, 120(15):150502, 2018.

\bibitem{fannes1992finitely}
Mark Fannes, Bruno Nachtergaele, and Reinhard~F Werner.
\newblock Finitely correlated states on quantum spin chains.
\newblock {\em Communications in Mathematical Physics}, 144(3):443--490, 1992.

\bibitem{verstraete2008matrix}
Frank Verstraete, Valentin Murg, and J~Ignacio Cirac.
\newblock Matrix product states, projected entangled pair states, and
  variational renormalization group methods for quantum spin systems.
\newblock {\em Advances in Physics}, 57(2):143--224, 2008.

\bibitem{cramer2010efficient}
Marcus Cramer, Martin~B Plenio, Steven~T Flammia, Rolando Somma, David Gross,
  Stephen~D Bartlett, Olivier Landon-Cardinal, David Poulin, and Yi-Kai Liu.
\newblock Efficient quantum state tomography.
\newblock {\em Nature communications}, 1:149, 2010.

\bibitem{gu2012quantum}
Mile Gu, Karoline Wiesner, Elisabeth Rieper, and Vlatko Vedral.
\newblock Quantum mechanics can reduce the complexity of classical models.
\newblock {\em Nature communications}, 3:762, 2012.

\bibitem{monras2016quantum}
Alex Monras and Andreas Winter.
\newblock Quantum learning of classical stochastic processes: The completely
  positive realization problem.
\newblock {\em Journal of Mathematical Physics}, 57(1):015219, 2016.

\bibitem{thompson2017using}
Jayne Thompson, Andrew~JP Garner, Vlatko Vedral, and Mile Gu.
\newblock Using quantum theory to simplify input--output processes.
\newblock {\em npj Quantum Information}, 3(6):1--8, 2017.

\bibitem{shi2006classical}
Y-Y Shi, L-M Duan, and Guifre Vidal.
\newblock Classical simulation of quantum many-body systems with a tree tensor
  network.
\newblock {\em Physical review a}, 74(2):022320, 2006.

\bibitem{vidal2008class}
Guifr{\'e} Vidal.
\newblock Class of quantum many-body states that can be efficiently simulated.
\newblock {\em Physical Review Letters}, 101(11):110501, 2008.

\bibitem{bai2020quantum}
Ge~Bai, Yuxiang Yang, and Giulio Chiribella.
\newblock Quantum compression of tensor network states.
\newblock {\em New Journal of Physics}, 22(4):043015, 2020.

\bibitem{coles2019strong}
Patrick~J Coles, M~Cerezo, and Lukasz Cincio.
\newblock Strong bound between trace distance and hilbert-schmidt distance for
  low-rank states.
\newblock {\em Physical Review A}, 100(2):022103, 2019.

\bibitem{buadescu2019quantum}
Costin B{\u{a}}descu, Ryan O'Donnell, and John Wright.
\newblock Quantum state certification.
\newblock In {\em Proceedings of the 51st Annual ACM SIGACT Symposium on Theory
  of Computing}, pages 503--514, 2019.

\bibitem{prugovevcki1977information}
E~Prugove{\v{c}}ki.
\newblock Information-theoretical aspects of quantum measurement.
\newblock {\em International Journal of Theoretical Physics}, 16(5):321--331,
  1977.

\bibitem{renes2004symmetric}
Joseph~M Renes, Robin Blume-Kohout, Andrew~J Scott, and Carlton~M Caves.
\newblock Symmetric informationally complete quantum measurements.
\newblock {\em Journal of Mathematical Physics}, 45(6):2171--2180, 2004.

\bibitem{plesch2011quantum}
Martin Plesch and {\v{C}}aslav Brukner.
\newblock Quantum-state preparation with universal gate decompositions.
\newblock {\em Physical Review A}, 83(3):032302, 2011.

\bibitem{peres2006quantum}
Asher Peres.
\newblock {\em Quantum theory: concepts and methods}, volume~57.
\newblock Springer Science \& Business Media, 2006.

\bibitem{mottonen2004quantum}
Mikko M{\"o}tt{\"o}nen, Juha~J Vartiainen, Ville Bergholm, and Martti~M
  Salomaa.
\newblock Quantum circuits for general multiqubit gates.
\newblock {\em Physical review letters}, 93(13):130502, 2004.

\bibitem{o2015quantum}
Ryan O'Donnell and John Wright.
\newblock Quantum spectrum testing.
\newblock In {\em Proceedings of the forty-seventh annual ACM symposium on
  Theory of computing}, pages 529--538, 2015.

\end{thebibliography}



\titleformat{\section}{\normalfont\bfseries\center\uppercase}{}{1em}{}
\titleformat{\subsection}{\normalfont\bfseries\center}{Supplementary Note \thesubsection}{1em}{}

\renewcommand{\thesection}{\hspace{0em}}
\renewcommand{\thesubsection}{\arabic{subsection}}
\renewcommand{\thesubsubsection}{\alph{subsubsection}}
\renewcommand{\subsectionautorefname}{Supplementary Note\hspace{-0.5ex}}
\renewcommand{\subsubsectionautorefname}{Supplementary Note\hspace{-0.5ex}}

\renewcommand{\sectionmark}[1]{ \markboth{\MakeUppercase{#1}}{} }


\section{Supplementary Notes}

\subsection{Preliminaries} \label{app:pre}
In this section, we give the theoretical background of our algorithms.

In the main text, we showed that independence tests can be used to find the last tooth of a comb. Here we justify this idea by the following proposition that gives a sufficient and necessary condition for $(A_x,B_y)$ to be the last tooth. The condition is based on the Choi state \cite{choi1975completely} of a channel, defined as $C\in\spc{H}_{A_1}\otimes\dots\otimes\spc{H}_{A_n}\otimes\spc{H}_{B_1}\otimes\dots\otimes\spc{H}_{B_n}$, $C := (\map{C}\otimes\map{I})(\ketbra{\Phi_+}/d_{A_1,...,A_n})$ where $\ket{\Phi_+} := \sum_{i=1}^{d_{A_1,...,A_n}} \ket{i}_{A_1\dots A_n}\ket{i}_{A_1\dots A_n}$ is a maximally entangled state, $d_{A_1,...,A_n}:=\prod_{i=1}^n d_{A_i}$ and $\{\ket{i}_{A_1\dots A_n}\}_{i=1}^{d_{A_1,...,A_n}}$ is a basis of $\spc{H}_{A_1}\otimes\dots\otimes\spc{H}_{A_n}$.
\begin{prop} \label{prop:last_tooth}  \cite{chiribella2008quantum}
    Let $C$ be the Choi state of a channel $\map{C}$. $(A_x,B_y)$ is the last tooth of $\map{C}$ in some causal unravelling of $\map{C}$ if and only if
    \begin{align}\label{eq:last_tooth}
       C_{A_1,\dots,A_n,B_{\neq y}} = C_{A_{\neq x},B_{\neq y}}  \otimes \frac{I_{A_x}}{d_{A_x}}
    \end{align}
    where $C_{A_1,\dots,A_n,B_{\neq y}}:=\Tr_{B_y}[C],C_{A_{\neq x},B_{\neq y}}:=\Tr_{A_xB_y}[C]$, and $I_{A_x}$ is the identity operator on system $A_{x}$.
\end{prop}

The condition Eq. (\ref{eq:last_tooth}) for the Choi state $C$ is equivalent to that for the process $\map{C}$, the sets of systems $\{A_x\}$ and $\{A_i| i \neq x\}\cup\{B_j |j \neq y\}$ are independent.

When we collect classical statistics from a quantum state or process, we need to ensure that the information encoded in the quantum data are preserved, for which purpose we introduce the informationally complete POVM \cite{prugovevcki1977information}.

\begin{definition} \label{def:IC}
    A POVM $\{P_\alpha\}$ on Hilbert space $\spc{H}$ is informationally complete if its elements form a spanning set of linear operators on $\spc{H}$, namely for any $X\in L(\spc{H})$, there exist complex numbers $\{p_\alpha\}$ such that $X=\sum_\alpha p_\alpha P_\alpha$.

    We say a set of states $\{\psi_\alpha\}$ is informationally complete if there exists an informationally complete POVM $\{P_\alpha\}$ such that ${\psi_\alpha} = P_\alpha/\Tr[P_\alpha], \forall \alpha$.
\end{definition}
For efficiency considerations, we prefer to choose an informationally complete POVM with the minimal number of elements. To form a spanning set of $L(\spc{H})$, the minimal informationally complete POVM contains $(\dim\spc{H})^2$ elements, which are linearly independent. 
A typical example of such POVMs is the symmetric informationally complete POVM (SIC-POVM) \cite{renes2004symmetric}, which has been found for most relevant low-dimensional Hilbert spaces.

The frame operators defined below will be used to characterize the sensitivity of measurement outcomes about the state being measured.
\begin{definition} \label{def:frame}
    The frame operator of an informationally complete POVM $\{P_\alpha\}$ on Hilbert space $\spc{H}$ is defined as $F :=\sum_\alpha\kketbbra{P_\alpha}\in L(\spc{H}\otimes\spc{H})$, where $\kket{P_\alpha} := (P_\alpha \otimes I) \kket{I} = \sum_{i,j}(P_\alpha)_{ij}\ket{i}\ket{j} \in \spc{H}\otimes\spc{H}$, where $\kket{I}:=\sum_i\ket{i}\ket{i}$ is the unnormalized maximally entangled state.

    For an informationally complete set of states $\{\psi_\alpha\}$, the frame operator is defined as $F':=\sum_\alpha\kketbbra{\psi_\alpha}$ with $\kket{\psi_\alpha}:= ({\psi_\alpha} \otimes I) \kket{I}$. 
\end{definition}

A frame operator with a larger minimum eigenvalue indicates that the outcome probabilities are more sensitive to any change of the state being measured, and the POVM is more efficient for obtaining the full information of the state. We will see the consequences of the minimum eigenvalue of frame operators in the analysis of our algorithms.


\subsection{Proof of Theorem \ref{thm:unitary2}} \label{app:unitary2}

We prove \autoref{thm:unitary2} in this section. 
\red{In \autoref{appss:indep}, we analyze the accuracy of the independence tests in the approximate case. In \autoref{appss:rank} we analyze how the error in each recursion step affects the final error of the algorithm. The main proof of \autoref{thm:unitary2} is in \autoref{appss:proof1}.}

We assume that the process $\map{C}$ is already in its standard form with all input and output wires having dimension $d_A$, and thus $r_{\map{C}}=\rank(\map{C})$.

\red{
\subsubsection{Accuracy of independence tests} \label{appss:indep}
}

\red{
We first consider the accuracy of the independence test implemented in \autoref{func:findlast_direct}, in an arbitrary iteration of \autoref{func:recursive}. Define $C_1 := \Tr_{B_y}[C]$ and $C_2 := I_{A_x}/d_{A_x} \otimes \Tr_{A_xB_y}[C]$ as in \autoref{func:findlast_direct}. By \autoref{lem:SWAP}, for each of the estimate $p_i, (i=1,2,3)$, with probability $1-\kappa$, $p_i$ is $\varepsilon$-close to the true value. By the union bound, with probability at least $1-3\kappa$, all three estimates are $\varepsilon$-close to their corresponding true values, and therefore
\begin{align} \label{eq:p1p2p3_dist}
 &~ \left|p_1+p_2-2p_3 - \|C_1- C_2\|_2^2 \right| \nonumber \\
\leq &~ \left|p_1- \Tr[C_1^2] \right| +\left|p_2- \Tr[C_2^2] \right| + \left|2p_3- 2\Tr[C_1C_2] \right| \nonumber \\
\leq &~ 4\varepsilon \,.
\end{align}

Next, using Eq. (\ref{eq:norm12_main}), 
\begin{align} \label{eq:norm12_C1C2}
    \sqrt{2} \|C_1- C_2\|_2 \leq \|C_1- C_2\|_1 \leq \sqrt{\frac{4\rank(C_1)\rank(C_2)}{\rank(C_1) + \rank(C_2)}} ~ \|C_1- C_2\|_2 \leq \sqrt{4\rank(C_1)} ~ \|C_1- C_2\|_2 
\end{align}
The rank of $C_1$ can be bounded as
\begin{align}
    & \rank(C_1) = \rank(\Tr_{B_y}[C]) \leq d_{B_y}\rank(C) = d_A \rank(\map{C}) 
\end{align}

Therefore by Eq. (\ref{eq:norm12_C1C2}),
\begin{align} \label{eq:norm12_C1C2_2}
    \sqrt{2} \|C_1- C_2\|_2 \leq \|C_1- C_2\|_1 \leq \sqrt{4 d_A \rank(\map{C})} ~ \|C_1- C_2\|_2
\end{align}

Now we consider the output of \autoref{func:findlast_direct}. If \autoref{func:findlast_direct} returns \True, then $p_1+p_2-2p_3\leq \delta$. From Eq. (\ref{eq:p1p2p3_dist}), $\|C_1- C_2\|_2 \leq \sqrt{\delta + 4\varepsilon}$ and then from Eq. (\ref{eq:norm12_C1C2_2}), $\|C_1- C_2\|_1 \leq \sqrt{4 d_A \rank(\map{C})} ~ \|C_1- C_2\|_2 \leq \sqrt{4 d_{A} \rank(\map{C}) (\delta + 4\varepsilon) }$. If \autoref{func:findlast_direct} returns \False, then $p_1+p_2-2p_3 >\delta$, and thus $\|C_1- C_2\|_1 \geq \sqrt{2} ~ \|C_1- C_2\|_2 > \sqrt{2 (\delta - 4\varepsilon)} $.
Since $\chione C {A_x}{A_{\neq x}B_{\neq y}} = \|C_1- C_2\|_1$, we have the following lemma:

\begin{lem}\label{lem:2norm_chi_direct}
Let $\map{C}$ be the input channel of \autoref{func:findlast_direct}. With probability $1-3\kappa$, \autoref{func:findlast_direct} satisfies the following:
\begin{enumerate}
\item If \autoref{func:findlast_direct} returns \True, then
\begin{align}
    \chione C {A_x}{A_{\neq x}B_{\neq y}} \leq \sqrt{4 d_{A} \rank(\map{C}) \, (\delta + 4\varepsilon)}\,, \label{eq:2norm_chi}
\end{align}
where $\rank(\map{C})$ is the Kraus rank of $\map{C}$;
\item If \autoref{func:findlast_direct} returns \False, then
\begin{align}
    \chione C {A_x}{A_{\neq x}B_{\neq y}} >\sqrt{2 (\delta - 4\varepsilon) }\,.
\end{align}
\end{enumerate}
\end{lem}

}

\subsubsection{Bounding the rank during recursion} \label{appss:rank}

In the following, we first give some lemmas on the Kraus rank of channels, 
and then prove \autoref{lem:rank2}, which can be used in the inductive proof of \autoref{thm:unitary2} to upper bound the rank of $\map{C}$ in any recursion.

\begin{lem}\label{lem:rank}
For a constant channel $\map{C}: S(\spc{H}_{A})\to S(\spc{H}_B)$ satisfying $\map{C}(X) = \Tr[X] \sigma_B$, one has $\rank(\sigma_B) =\rank(\map{C})/ \dim\spc{H}_A$, where $\rank(\map{C})$ is the Kraus rank of channel $\map{C}$.
\end{lem}
\begin{proof}
   Since $\map{C}(X) = \Tr[X] \sigma_B$, the Choi state of $\map{C}$ is $C = I_A/d_A \otimes \sigma_B$. Then $\rank(C) = \rank(I_A)\rank(\sigma_B)$, and $\rank(\sigma_B) = \rank(C) / \rank(I_A) = \rank(\map{C})/ \dim\spc{H}_A$.
\end{proof}
\begin{cor}\label{cor:comb_rank}
If $\map{C}\in\set{Comb}[(A_1,B_1),(A_2,B_2)]$, then $\rank(\map{C}_{A_1\to B_1}) \leq \rank(\map{C}) d_{B_2}/d_{A_2}$\red{, where $\map{C}_{A_1\to B_1}$ is the channel obtained from $\map{C}$ by inputting the maximally mixed state to $A_2$ and discarding the output of $B_2$.}
\end{cor}
\begin{proof}
    \red{Define $\map{C}_{A_2\to A_1B_1}$ as $\map{C}_{A_2\to A_1B_1}(\rho) := \Tr_{B_2}[ (\map{C}\otimes\map{I}_{A_1'})(\rho \otimes \kketbbra{I}_{A_1A_1'}/d_{A_1})]$, namely the map with Choi state equal to $C_{A_1B_1A_2}:=\Tr_{B_2}[C]$, having $A_2$ as the input and $A_1,B_1$ as outputs. $\map{C}\in\set{Comb}[(A_1,B_1),(A_2,B_2)]$ indicates that $C_{A_1B_1A_2} = C_{A_1B_1}\otimes C_{A_2}$, and thus $\map{C}_{A_2\to A_1B_1}$ is a constant channel: $\map{C}_{A_2\to A_1B_1}(\rho) = \Tr[\rho] C_{A_1B_1}$. By \autoref{lem:rank}, $\rank(C_{A_1B_1}) =\rank(\map{C}_{A_2\to A_1B_1})/ d_{A_2}$.} Then
    \begin{align}
        \rank(\map{C}_{A_1\to B_1}) = \rank(C_{A_1B_1}) = \rank(\map{C}_{A_2\to A_1B_1})/ d_{A_2} = \rank(\Tr_{B_2}[C])/ d_{A_2} \leq  \rank(\map{C}) d_{B_2}/d_{A_2}
    \end{align}
\end{proof}

\begin{lem}\label{lem:rank2}
Let $\map{C}^{(n)}$ be the initial value of $\map{C}$ in \autoref{func:recursive}, and let $\map{C}^{(k)}$, which has $k$ input wires and $k$ output wires, be the value of $\map{C}$ in the $(n-k+1)$-th recursive call of \autoref{func:recursive}. Define the follow propositions:
\begin{itemize}
    \item $P(k)$: in the $(n-k+1)$-th recursive call of \autoref{func:recursive}, the pair $(A_x,B_y)$ passing the test is indeed the last tooth, namely $\map{C}^{(k)} \in \set{Comb}[(A_{\neq x},B_{\neq y}), (A_x,B_y)]$, where $A_{\neq x}$ ($B_{\neq y}$) denotes all input (output) wires of $\map{C}^{(k)}$ excluding $A_x$ ($B_y$).
    \item $Q(k)$: $\rank(\map{C}^{(k)}) \leq \rank(\map{C}^{(n)})$.
\end{itemize}
Then, under \autoref{ass:dAdM}, $Q(k) \wedge P(k) \Rightarrow Q(k-1)$, for every $k=n,n-1,\dots,2$.
\end{lem}
\begin{proof}
    

    According to $P(k)$, $\map{C}^{(k)} \in \set{Comb}[(A_{\neq x},B_{\neq y}), (A_x,B_y)]$, and the following channel $\map{C}^{(k)}_{A_x\to A_{\neq x}B_{\neq y}}: \spc{H}_{A_x} \to \spc{H}_{B_{\neq y}} \otimes \spc{H}_{A'_{\neq x}}$ is a constant channel:
    $$\map{C}^{(k)}_{A_x\to A_{\neq x}B_{\neq y}}(\rho_{A_x}) := \Tr_{B_y} \left[ \left( \map{C}^{(k)}\otimes\map{I_{A'_{\neq x}}} \right) \left( \bigotimes_{z \neq x} \frac{\kketbbra{I}_{A_zA'_z}}{d_{A_z}} \otimes \rho_{A_x} \right) \right]$$

    Since $\map{C}^{(k)}_{A_x\to A_{\neq x}B_{\neq y}}$ is a constant channel, $\map{C}^{(k)}_{A_x\to A_{\neq x}B_{\neq y}}(X) = \Tr[X] C^{(k-1)}$, where $C^{(k-1)}$ equals to the Choi state of the updated channel $\map{C}^{(k-1)}$. According to \autoref{lem:rank}, $\rank(\map{C}^{(k-1)})\leq \rank(\map{C}^{(k)})\dim\spc{H}_{B_y}/ \dim\spc{H}_{A_x} =  \rank(\map{C}^{(k)})$, where we used $\dim\spc{H}_{A_x} = \dim\spc{H}_{B_y} = d_A$ since we have assumed that the process is in its standard form in \autoref{ass:dAdM}. 
    By $Q(k)$, $\rank(\map{C}^{(k)}) \leq \rank(\map{C}^{(n)})$, thus $\rank(\map{C}^{(k-1)}) \leq \rank(\map{C}^{(k)}) \leq \rank(\map{C}^{(n)})$ and $Q(k-1)$ holds.
\end{proof}

\subsubsection{Proof of \autoref{thm:unitary2}} \label{appss:proof1}

Now we are ready to prove \autoref{thm:unitary2}.

\red{
    
\begin{proof}[Proof of \autoref{thm:unitary2}]

In \autoref{func:recursive}, using the relation $T_{\rm test}(n) \leq n^2 + T_{\rm test}(n-1)$, we can find that there are $T_{\rm test}(n) \leq n^3$ independence tests, each of which invokes three SWAP tests. Each SWAP test produces an estimate within error $\varepsilon$ with probability no less than $1-\kappa$ according to \autoref{lem:SWAP}.
Let $\kappa = \kappa_0/3n^3$, and from now we assume all SWAP tests are within error $\varepsilon$, which has probability no less than $1-3n^3\kappa = 1-\kappa_0$.

Let $\map{C}^{(n)}$ be the initial value of $\map{C}$ in \autoref{func:recursive}, and let $\map{C}^{(k)}$, which has $k$ input wires and $k$ output wires, be the value of $\map{C}$ in the $(n-k+1)$-th recursive call of \autoref{func:recursive}.

Consider testing whether $(A_x,B_y)$ is the last tooth of $\map{C}^{(k)}$ with \autoref{func:findlast_direct}. Let $S$ be the set of all wires of $\map{C}^{(k)}$ excluding $A_x$ and $B_y$. If $(A_x,B_y)$ passes the test, namely \autoref{func:findlast_direct} returns \True, by \autoref{lem:2norm_chi_direct}, 
\begin{align} \label{eq:chione_deltaepsi}
    \chione {C^{(k)}}{A_x}{S} \leq \sqrt{4 d_{A} \rank(\map{C}^{(k)})(\delta + 4\varepsilon) } \, .
\end{align}

Now we pick $\delta \leq \frac{\chi_{\min}^2}{8 d_A \rank(\map{C}^{(k)})}$ and $\varepsilon = \delta/5$. Eq. (\ref{eq:chione_deltaepsi}) then indicates $\chione {C^{(k)}}{A_x}{S} \leq \sqrt{4 d_{A} \rank(\map{C}^{(k)}) (\delta + 4\varepsilon)} < \chi_{\min}$.

Since $\chione {C^{(k)}}{A_x}{S} = \chione {C^{(n)}}{A_x}{S}$,  by \autoref{ass:ST_threshold}, $\chione {C^{(k)}}{A_x}{S}=0$. Therefore, if \autoref{func:findlast_direct} returns \True, then $(A_x,B_y)$ must be the last tooth of the comb, in some causal unravelling of $\map{C^{(k)}}$. Using the notation in \autoref{lem:rank2}, $P(k)$ holds for this iteration. On the other hand, if \autoref{func:findlast_direct} returns \False,  by \autoref{lem:2norm_chi_direct}, 
\begin{align}
    \chione {C^{(n)}}{A_x}{S} = \chione {C^{(k)}}{A_x}{S} > \sqrt{2\, (\delta - 4\varepsilon)} > 0 \,,
\end{align}
and $(A_x,B_y)$ does not satisfy the condition for being last tooth of the comb.

Let $r_{\map{C}}:= \rank(\map{C}^{(n)})$. Take the definition of $P(k)$ and $Q(k)$ from \autoref{lem:rank2}. Now we perform an induction to prove $P(k)$ and $Q(k)$ for every $k$.

First, $Q(n)$ obviously holds according to \autoref{ass:dAdM}. Now, we assume that $Q(k)$ holds, and want to show that $Q(k-1)$ also holds.

By $Q(k)$, we have $\rank(\map{C}^{(k)}) \leq r_{\map{C}}$.
In any recursive call of \autoref{func:recursive}, if $Q(k)$ holds, then we can choose $\delta=\frac{\chi_{\min}^2}{8 d_A r_{\map{C}}} \leq \frac{\chi_{\min}^2}{8 d_A \rank(\map{C}^{(k)})}$, and by the argument above, $P(k)$ holds. Combining this with \autoref{lem:rank2}, we have $Q(k) \Rightarrow P(k)$ and $Q(k) \wedge P(k) \Rightarrow Q(k-1)$, meaning $Q(k) \Rightarrow Q(k-1)$.

This completes the induction, and we conclude that $Q(k)$ and $P(k)$ holds for every $k$. $P(k)$ being true for every $k$ shows the correctness of our algorithm. 


According to \autoref{func:SWAPTEST}, $N=O(\varepsilon^{-2} \log\kappa^{-1}) = O(\delta^{-2} \log (n \kappa_0^{-1})) = O(d_A^2 r_{\map{C}}^2 \chi_{\min}^{-4} \log (n \kappa_0^{-1}))$. The sample complexity is then
\begin{align}
    T_{\rm sample} = N T_{\rm test}(n)  = O\left(n^3 d_A^2 r_{\map{C}}^2 \chi_{\min}^{-4} \log (n \kappa_0^{-1}) \right)
\end{align}
and the computational complexity is $O(T_{\rm sample} n \log d_A)$ from the realization of SWAP tests.


\end{proof}
}

\subsection{Independence test algorithm between all input-output pairs} \label{app:all_pairs}
In this section, we show the details of the algorithm to test 
the independence between every pair of input and output wires, resulting in a Boolean matrix $\ind_{ij}$. 
Our independence test algorithm will be based on the estimation of $\chi_1(\rho_{A,B})$. The idea is essentially quantum state tomography: measure $\rho_{AB}$ with informationally complete POVMs $\{P_\alpha\}$ on system $A$ and $\{Q_\beta\}$ on system $B$, use the outcome statistics to reconstruct $\rho_{AB}$, and estimate $\chi_1(\rho_{A,B})$ based on the reconstructed state. The algorithm is given in \autoref{alg:c1}.

\begin{function}[H]
    \caption{estimatechi1(\pbox{5cm}{$\{(a^{(k)},b^{(k)})\}_{k=1}^{N},\{P_\alpha\},\{Q_\beta\}$})\label{alg:c1}}
    \Indentp{0.5em}
    \SetInd{0.5em}{1em}
    \SetKwInOut{Input}{Input}
    \SetKwInOut{Output}{Output}
    \SetKw{Break}{break}
    \SetKwData{matched}{matched}

    \ResetInOut{Input}
    \Input{$N$ measurement outcomes $\{(a^{(k)},b^{(k)})\}_{k=1}^{N}$ of POVM $\{P_\alpha \otimes Q_\beta\}$ on state $\rho_{AB}$}
    \ResetInOut{Output}
    \Output{An estimate of $\chi_1(\rho_{A,B})$} 
    \BlankLine

    Initialize $\hat{p}_{\alpha\beta}$ to zero for all $\alpha=1,\dots,d_A^2,$ and $\beta=1,\dots,d_B^2$\;
    \For{$k \gets 1$ \KwTo $N$}
    {
            $\hat{p}_{a^{(k)}b^{(k)}} \gets \hat{p}_{a^{(k)}b^{(k)}} + 1/ N$\;
    }

    $F \gets \sum_\alpha \kketbbra{P_\alpha}$\;
    $G \gets \sum_\beta \kketbbra{Q_\beta}$\;
    $S \gets \sum_\alpha \ket{\alpha}\bbra{P_\alpha}$, where $\{\ket\alpha\}$ is a set of orthonormal vectors indexed by $\alpha$\;
    $R \gets \sum_\beta \ket{\beta}\bbra{Q_\beta}$, where $\{\ket\beta\}$ is a set of orthonormal vectors indexed by $\beta$\;
    $\ket{\hat p} \gets \sum_{\alpha,\beta} \hat p_{\alpha\beta}\ket\alpha \ket\beta$\;
    Compute operator $\hat\rho_{AB}$ by $\kket{\hat{\rho}_{AB}} = ( F^{-1} S^\dag \otimes G^{-1} R^\dag)\ket{\hat p}$ \;

    Return $\chi_1(\hat\rho_{A,B})$\;
\end{function}

Now we analyze the accuracy and complexity of \autoref{alg:c1}.
Define the constant $\xi:=
\frac{ \sqrt{\lambda_{\min}(F)\lambda_{\min}(G)}}
{ \sqrt{d_A^2 d_B^2 + 4d_B^2 + 4d_A^2 } d_Ad_B}$
where $\lambda_{\min}(X)$ denotes the minimum eigenvalue of operator $X$
and $F$ and $G$ are defined in \autoref{alg:c1}. Then we define
\begin{align}\label{eq:kappaepsilon}
    \kappa(\varepsilon)
    := 2 (d_A^2d_B^2+d_A^2+d_B^2) e^{-2 \xi^2 \varepsilon^2 N} .
\end{align}
The error of the estimate $\chi_1(\hat\rho_{A,B})$ is given in the following lemma, whose proof is in \autoref{appss:alg_ind}.
\begin{lem}\label{lem:alg_ind}
    For any positive real number $\varepsilon>0$, over the $N$ measurement outcomes on independent copies of $\rho_{AB}$, with probability $1-\kappa_0$, the output of \autoref{alg:c1} estimates $\chi_1(\rho_{A,B})$ with the following error bound
    \begin{align}\label{eq:c1_bound}
        |\chi_1(\rho_{A,B}) - \chi_1(\hat\rho_{A,B})| \leq \varepsilon \,.
    \end{align}
    That is, the inequality
    \begin{align}\label{eq:c1_bound-ie}
    \rm{Pr} (   |\chi_1(\rho_{A,B}) - \chi_1(\hat\rho_{A,B})| > \varepsilon )
    \le \kappa(\varepsilon)
    \end{align}
    holds for any positive real number $\varepsilon>0$.
    Here, $\rm{Pr} ( C)$ expresses the probability that the condition $C$ is satisfied.
\end{lem}

In other words, to reach confidence $1-\kappa_0$ and error $\varepsilon_0$, the number of samples $N$ is of order
\begin{align} \label{eq:alg_ind_N}
    N = O\left(d_A^4 d_B^4 \lambda_{\min}^{-1}(F) \lambda_{\min}^{-1}(G) \varepsilon_0^{-2} \log(d_A d_B \kappa_0^{-1}) \right) \,.
\end{align}
If one choose $\{P_\alpha\}$ and $\{Q_\beta\}$ to be SIC-POVMs, $\lambda_{\min}^{-1}(F) = d_A(d_A+1) = O(d_A^2)$ and $\lambda_{\min}^{-1}(G) = d_B(d_B+1) = O(d_B^2)$ \cite{renes2004symmetric}. In this case, the number of samples $N = O\left(d_A^6 d_B^6 \varepsilon_0^{-2} \log(d_A d_B \kappa_0^{-1}) \right)$ is polynomial in $d_A$ and $d_B$.

\autoref{alg:c1} could be used on quantum channels to discover the correlation between input and output wires. To characterize a quantum channel, we could transform it to its Choi state by feeding the input with half of a maximally entangled state, and apply \autoref{alg:c1} on the Choi state. To save the cost of creating and maintaining entangled states, we could use \autoref{alg:c1} in an alternative way. Let $\{\psi_\alpha\}_{\alpha=1}^{d_A^2}$ be an informationally complete set of states such that $\psi_\alpha = P^T_\alpha/\Tr[P_\alpha]$. One samples $\psi_\alpha$ from $\{\psi_\alpha\}_{\alpha=1}^{d_A^2}$ with probability $\Tr[P_\alpha]/d_A$, apply quantum channel $\map{C}: L(\spc{H}_A)\to L(\spc{H}_B)$ on $\psi_\alpha$, and measure the output with $\{Q_\beta\}_{\beta=1}^{d_B^2}$. This process generates the same statistics for $(\alpha,\beta)$ according to the following equation:
\begin{align}\label{eq:channel_statistics}
    \frac{1}{d_A}\Tr[P_\alpha]\Tr[Q_\beta \map{C}(\psi_\alpha)] = \frac{1}{d_A}\Tr[Q_\beta \map{C}(P^T_\alpha)] = \Tr\left[\frac{1}{d_A}\sum_{i,j} P_\alpha\ket{i}\bra{j} \otimes Q_\beta \map{C}(\ket{i}\bra{j})\right] = \Tr[(P_\alpha \otimes Q_\beta)C]
\end{align}
where $C := \frac{1}{d_A} \sum_{i,j} \ket{i}\bra{j} \otimes \map{C}(\ket{i}\bra{j})$ is the Choi state of channel $\map{C}$. The left hand side of Eq. (\ref{eq:channel_statistics}) is the probability of sampling $\psi_\alpha$ times the probability of outcome $\beta$ on the output conditioned on the input being $\psi_\alpha$, and the right hand side of Eq. (\ref{eq:channel_statistics}) is the probability of outcome $(\alpha,\beta)$ if one measure the Choi state $C$ with POVM $\{P_\alpha \otimes Q_\beta\}$.


Now we adopt \autoref{alg:c1} to causal unravelling. By testing the correlation between every pair of input and output wires, one is able to infer the causal unravellings: if output $B_j$ is correlated to input $A_i$, $B_j$ must be after $A_i$. In this section, we study cases where such reasoning yields the full causal unravelling.

We introduce an algorithm that discovers the correlations between every pair of input and output wires. The idea is to first collect all the required classical data, and deduce all the correlations according to the statistics. During this procedure, the data are reused for different input-output pairs, leading to fewer queries to the channel $\map{C}$.

\begin{function}[H]
    \caption{independence($\map{C}, N, \chi_-$)\label{alg:order_linear2}\label{alg:ind}}
    \Indentp{0.5em}
    \SetInd{0.5em}{1em}
    \SetKwInOut{Input}{Input}
    \SetKwInOut{Output}{Output}
    \SetKwFunction{estimatechi}{estimatechi1}

    \ResetInOut{Input}
    \Input{Black-box access of a quantum comb $\map{C}$ with input wires $A_1,\dots,A_n$ and output wires $B_1,\dots,B_n$}
    \ResetInOut{Output}
    \Output{An array $\ind$ where $\ind_{i,j}$ being \True{} indicates $A_i$ and $B_j$ are independent}
    \BlankLine

    Pick informationally complete linearly independent POVMs $\{P_{\alpha}^{i}\}_{\alpha=1}^{d_{A_i}^2}$ for each input $A_i$ and $\{Q_{\beta}^{j}\}_{\beta=1}^{d_{B_j}^2}$ for each output $B_j$\; 
    Define ${\psi_\alpha^{i}}:= (P_{\alpha}^{i})^T/\Tr[P_{\alpha}^{i}],~ \forall i=1,\dots,n,~\alpha=1,\dots,d_{A_i}^2$\;
    Generate $N$ random $n$-tuples $\{(a_1^{(k)},a_2^{(k)},\dots,a_n^{(k)})\}_{k=1}^{N}$ where $a_i^{(k)}$ is chosen from $\{1,\dots,d_{A_i}^2\}$ with probability $\Pr(a_i^{(k)}=\alpha ) = \Tr[P_{\alpha}^{i}]/d_{A_i}$\;
    \For{$k \gets 1$ \KwTo $N$}
    {
        Apply $\map{C}$ to input state $\psi^1_{a_1^{(k)}} \otimes \psi^2_{a_2^{(k)}} \dots \otimes \psi^n_{a_n^{(k)}}$ \;
        For each $j$, measure $j$-th output wire with $\{Q_{\beta}^{j}\}_{\beta=1}^{d_{B_j}^2}$. Let the outcomes be $(b_1^{(k)},b_2^{(k)},\dots,b_n^{(k)}), ~ b_j^{(k)} \in \{1,\dots,d_{B_j}^2\}$ \;
    }
    \For{$i \gets 1$ \KwTo $n$}
    {
        \For{$j \gets 1$ \KwTo $n$}
        {
            \uIf{$\estimatechi(\{(a_i^{(k)},b_j^{(k)})\}_{k=1}^{N}, \{ P_\alpha^i \}, \{ Q_\beta^j \}) \leq \chi_-$}
            {
                $\ind_{i,j} \gets \True$ \;
            }
            \Else
            {
                $\ind_{i,j} \gets \False$ \;
            }
        }
    }
    Output $\ind$\;
\end{function}

To understand \autoref{alg:ind}, we observe that $\{(a_i^{(k)},b_j^{(k)})\}_{k=1}^{N}$ are sampled from the same distribution as the outcomes of measuring $C_{A_i,B_j}$ with $\{P_\alpha^i \otimes Q_\beta^j\}$, and $\FuncSty{estimatechi1}$ produces an estimate of $\chione{C}{A_i}{B_j}$. Conditioned on that all $n^2$ number of calls to $\FuncSty{estimatechi1}$ succeed, which has probability no less than $1-n^2\kappa_0$, all the estimations $\chione{C}{A_i}{B_j}$ have error no greater than $\varepsilon_0$. We summarize these observations in the following lemma:

\begin{lem}\label{lem:ind}
With probability $1-n^2\kappa_0$, \autoref{alg:ind} produces an output satisfying:
\begin{enumerate}
    \item if $\ind_{i,j} = \False$, then $\chi_1(C_{A_i,B_j}) > \chi_- -\varepsilon_0$
    \item if $\ind_{i,j} = \True$, then $\chi_1(C_{A_i,B_j}) \leq \chi_- + \varepsilon_0$
\end{enumerate}
where $\kappa_0$ and $\varepsilon_0$ are related as
\begin{align} \label{eq:kappaepsilon2}
\kappa_0 = 2 (d_A^2d_B^2+d_A^2+d_B^2) e^{-2 \xi^2 \varepsilon_0^2 N} .
\end{align}
where $\xi:=\frac{\sqrt{\lambda_{\min,\rm in}\lambda_{\min,\rm out}}}{ \sqrt{d_A^2 d_B^2 + 4d_B^2 + 4d_A^2 } d_Ad_B}$,  $\lambda_{\min,\rm in}:=\min{\lambda_{\min}(F_i)}$, $\lambda_{\min,\rm out}:=\min{\lambda_{\min}(G_i)}$, and $F_i$ ($G_i$) is the frame opertor of $\{P_{\alpha}^{i}\}_{\alpha=1}^{d_{A_i}^2}$ ($\{Q_{\beta}^{j}\}_{\beta=1}^{d_{B_j}^2}$).
\end{lem}

\autoref{lem:ind_informal} is a paraphrase of \autoref{lem:ind} by rewriting Eq. (\ref{eq:kappaepsilon2}) and choosing $\{P_{\alpha}^{i}\}_{\alpha=1}^{d_{A_i}^2}$ and $\{Q_{\beta}^{j}\}_{\beta=1}^{d_{B_j}^2}$ to be SIC-POVMs so that $\lambda_{\min,\rm in} = O(d_A^2)$ and $\lambda_{\min,\rm out} = O(d_B^2)$.

\subsubsection{Proof of Lemma \ref{lem:alg_ind}} \label{appss:alg_ind}


\begin{lem} \label{lem:2normbounds}
    Let $\{P_\alpha\}$ be an informationally complete POVM and $F :=\sum_\alpha \kketbbra{P_\alpha}$ be its frame operator. For two Hermitian operators $\rho$ and $\sigma$, define $p_\alpha:=\Tr[P_\alpha \rho]$ and $q_\alpha:=\Tr[P_\alpha \sigma]$. Then one has
    \begin{align} \label{eq:2normbounds}
        \frac{\sum_{\alpha}(p_\alpha-q_\alpha)^2}{\lambda_{\max}(F)} \leq \|\rho - \sigma\|_2^2 \leq \frac{\sum_{\alpha}(p_\alpha-q_\alpha)^2}{\lambda_{\min}(F)} \,,
    \end{align}
    where $\|X\|_2:= \sqrt{\Tr[X^\dag X]}$ denotes the Schatten 2-norm, also known as the Frobenius norm, and $\lambda_{\max}(F)$ and $\lambda_{\min}(F)$ are the maximum and minimum eigenvalues of $F$, respectively.
\end{lem}
\begin{proof}
    Let $T := \rho - \sigma$ and $t_{\alpha} := p_\alpha - q_\alpha$. Then $t_\alpha = \Tr[P_\alpha \rho]-\Tr[P_\alpha \sigma] = \Tr[P_\alpha T]$. Define the vector $\ket{t} := \sum_\alpha t_\alpha \ket{\alpha}$, and operator $S := \sum_\alpha \ket{\alpha}\bbra{P_\alpha}$, where $\{\ket\alpha\}$ is a set of orthonormal vectors indexed by $\alpha$. Then one has $\ket{t} = S\kket{T}$ and $S^\dag S = F$.
    \begin{align}
        \sum_\alpha t_\alpha^2 =  \braket{t}{t} = \bbra{T} S^\dag S \kket{T} = \bbra{T} F \kket{T}\,.
    \end{align}
    Since $\lambda_{\min}(F)\bbrakket{T}{T} \leq \bbra{T} F \kket{T} \leq \lambda_{\max}(F)\bbrakket{T}{T}$, we obtain
    \begin{align}
        \lambda_{\min}(F)\bbrakket{T}{T} \leq \sum_\alpha t_\alpha^2 \leq \lambda_{\max}(F)\bbrakket{T}{T} \,,
    \end{align}
    which is equivalent to Eq. (\ref{eq:2normbounds}) since $\bbrakket{T}{T} = \|T\|_2^2$.
\end{proof}

\begin{proof}[Proof of \autoref{lem:alg_ind}]
Let $\hat\rho_A := \Tr_B[\hat\rho_{AB}]$ and $\hat\rho_B := \Tr_A[\hat\rho_{AB}]$ be the marginal states of the reconstructed state $\hat\rho_{AB}$.
Define $\tau := \rho_{AB} - \rho_A \otimes \rho_B$ and $\hat\tau := \hat\rho_{AB} - \hat\rho_A \otimes \hat\rho_B$, the error of \autoref{alg:c1} is then the difference between $\|\tau\|_1=\chione\rho{A}{B}$ and $\|\hat{\tau}\|_1=\chione{\hat\rho}{A}{B}$. 

Let $p_{\alpha\beta}:= \Tr[(P_\alpha \otimes Q_\beta)\rho_{AB}]$, $p^A_{\alpha}:=\sum_\beta p_{\alpha\beta}$ and $p^B_{\beta} = \sum_\alpha p_{\alpha\beta}$. Let $\hat{p}^A_{\alpha}:=\sum_\beta \hat{p}_{\alpha\beta}$ and $\hat{p}^B_{\beta} = \sum_\alpha \hat{p}_{\alpha\beta}$. Over the $N$ measurement outcomes, $\hat{p}_{\alpha\beta}$ is the average of independent Bernouli random variables, and has mean $p_{\alpha\beta}$. The same is true for $\hat{p}^A_{\alpha}$ and $\hat{p}^B_{\beta}$, whose means are $p^A_{\alpha}$ and $p^B_{\beta}$, respectively.

Let $\varepsilon_1:= \xi\varepsilon$. By Hoeffding's inequality,
\begin{align}
    \Pr(|p_{\alpha\beta} - \hat{p}_{\alpha\beta}| \leq \varepsilon_1) &\geq 1-2 e^{-2\varepsilon_1^2 N} \\
    \Pr(|p^A_{\alpha} - \hat{p}^A_{\alpha}| \leq \varepsilon_1) &\geq 1-2 e^{-2\varepsilon_1^2 N}\\
    \Pr(|p^B_{\beta} - \hat{p}^B_{\beta}| \leq \varepsilon_1) &\geq 1-2 e^{-2\varepsilon_1^2 N}
\end{align}
By the union bound,
\begin{align}
    \Pr\left[ \left(\forall\alpha,\beta, |p_{\alpha\beta} - \hat{p}_{\alpha\beta}| \leq \varepsilon_1\right) \wedge \left(\forall \alpha, |p^A_{\alpha} - \hat{p}^A_{\alpha}| \leq  \varepsilon_1\right) \wedge \left(\forall \beta, |p^B_{\beta} - \hat{p}^B_{\beta}| \leq \varepsilon_1\right)  \right] \geq 1-2 (d_A^2d_B^2+d_A^2+d_B^2) e^{-2\varepsilon_1^2 N},
\end{align}
namely with probability no less than $1-2 (d_A^2d_B^2+d_A^2+d_B^2) e^{-2\varepsilon_1^2 N} = 1-\kappa(\varepsilon)$, all $(d_A^2d_B^2+d_A^2+d_B^2)$ number of the following inequalities
\begin{align}
    |p_{\alpha\beta} - \hat{p}_{\alpha\beta}| \leq & \varepsilon_1, ~\forall \alpha,\beta \label{eq:phat1} \\
    |p^A_{\alpha} - \hat{p}^A_{\alpha}| \leq & \varepsilon_1, ~\forall \alpha\\
    |p^B_{\beta} - \hat{p}^B_{\beta}| \leq & \varepsilon_1, ~\forall \beta \label{eq:phat3}
\end{align}
are satisfied. We then analyze the estimation error of $\chione\rho{A}{B}$ in this case. Define $t_{\alpha\beta}:=p_{\alpha\beta} - p^A_{\alpha}p^B_{\beta}$. By Eqs. (\ref{eq:phat1}-\ref{eq:phat3}),
\begin{align}
|t_{\alpha\beta} - \hat{t}_{\alpha\beta}|  & = |p_{\alpha\beta} - \hat{p}_{\alpha\beta} - p^A_{\alpha}p^B_{\beta} + \hat{p}^A_{\alpha}\hat{p}^B_{\beta}| \\
&\leq  |p_{\alpha\beta} - \hat{p}_{\alpha\beta}| + |p^A_{\alpha}p^B_{\beta} - p^A_{\alpha}\hat{p}^B_{\beta}|  + |p^A_{\alpha}\hat{p}^B_{\beta} - \hat{p}^A_{\alpha}\hat{p}^B_{\beta}| \\
&\leq  \varepsilon_1 + p^A_{\alpha}|p^B_{\beta} -\hat{p}^B_{\beta}|  + |p^A_{\alpha} - \hat{p}^A_{\alpha}| \hat{p}^B_{\beta} \\
&\leq  (1+p^A_{\alpha} + \hat{p}^B_{\beta})\varepsilon_1
\end{align}

Since $\sum_\alpha p^A_{\alpha} = \sum_\beta \hat{p}^B_{\beta} = 1$, we have
\begin{align}
\sum_{\alpha,\beta} (t_{\alpha\beta} - \hat{t}_{\alpha\beta})^2 &\leq \sum_{\alpha,\beta} (1+p^A_{\alpha} + \hat{p}^B_{\beta})^2\varepsilon_1^2\\
&\leq  \sum_{\alpha,\beta} (1+4p^A_{\alpha} + 4\hat{p}^B_{\beta}) \varepsilon_1^2\\
& = (d_A^2 d_B^2 + 4d_B^2 + 4d_A^2) \varepsilon_1^2
\end{align}

Note that $t_{\alpha\beta} = \Tr[(P_\alpha\otimes Q_\beta)\tau]$ and $\hat{t}_{\alpha\beta} = \Tr[(P_\alpha\otimes Q_\beta)\hat{\tau}]$. Applying \autoref{lem:2normbounds} to POVM $\{P_\alpha\otimes Q_\beta\}$ and operators $\tau$ and $\hat{\tau}$, we obtain
\begin{align}
\|\tau - \hat{\tau}\|_2^2 \leq \frac{\sum_{\alpha,\beta}\left(t_{\alpha\beta}-\hat{t}_{\alpha\beta}\right)^2}{\lambda_{\min}(F\otimes G)} \leq \frac{(d_A^2 d_B^2 + 4d_B^2 + 4d_A^2) \varepsilon_1^2}{\lambda_{\min}(F)\lambda_{\min}(G)}
\end{align}

Using the inequality $\|X\|_1^2 \leq \rank(X) \|X\|_2^2$ for any operator $X$ \cite{coles2019strong}, we have
\begin{align}
\|\tau - \hat{\tau}\|_1 \leq \sqrt{\rank(\tau - \hat{\tau})}\|\tau - \hat{\tau}\|_2 \leq d_Ad_B \sqrt{\frac{d_A^2 d_B^2 + 4d_B^2 + 4d_A^2 }{\lambda_{\min}(F)\lambda_{\min}(G)}} \, \varepsilon_1 = \varepsilon
\end{align}
which proves Eq. (\ref{eq:c1_bound}) since $| \|\tau\|_1 - \|\hat{\tau}\|_1 | \leq \|\tau - \hat\tau\|_1$.
\end{proof}

\subsection{Details of the algorithm under Assumption \ref{ass:totalorder}}\label{app:totalorder}

In this section, we give the pseudocode of the algorithm under \autoref{ass:totalorder} and discuss its correctness and efficiency.

\begin{algorithm}[H]
    \caption{Quantum causal unravelling algorithm under \autoref{ass:totalorder}}\label{alg:order_linear}
    \Indentp{0.5em}
    \SetInd{0.5em}{1em}
    \SetKwInOut{Preprocessing}{Preprocessing}
    \SetKwInOut{Input}{Input}
    \SetKwInOut{Output}{Output}
    \SetKwFor{Whenever}{whenever}{do}{end}

    \ResetInOut{Input}
    \Input{Black-box access of a quantum comb $\map{C}$ with input wires $A_1,\dots,A_n$ and output wires $B_1,\dots,B_n$, number of queries $N$}
    \ResetInOut{Output}
    \Output{Ordering of inputs and outputs} 
    \BlankLine

    $\ind \gets \independent(\map{C}, N, \chi_{\min}/2)$\;
    \For{$k \gets 1$ \KwTo $n$}
    {
        $c_A(k) \gets | \{ i | \ind_{k,i}=\False \}|$ \;
        $c_B(k) \gets | \{ i | \ind_{i,k}=\False \}|$ \;
    }
    Sort $\{A_1,\dots,A_n\}$ in descending order of $c_A(k)$ as $\{A_{\sigma(1)},\dots,A_{\sigma(n)}\}$\;
    Sort $\{B_1,\dots,B_n\}$ in ascending order of $c_B(k)$ as $\{B_{\pi(1)},\dots,B_{\pi(n)}\}$\;
    Output $(A_{\sigma(1)},B_{\pi(1)}),\dots,(A_{\sigma(n)},B_{\pi(n)})$\;
\end{algorithm}


Here the threshold $\chi_-$ is set to $\chi_-=\chi_{\min}/2$. Now we discuss the correctness and efficiency of \autoref{alg:order_linear} by proving \autoref{thm:order_linear}.


\begin{proof}[Proof of \autoref{thm:order_linear}]
    Let $\varepsilon_0 = \chi_{\min}/3$ and $\kappa_0 = \kappa / n^2$. By \autoref{lem:ind}, with probability $1-n^2 \kappa_0 = 1-\kappa$, the independence tests produce the true answers: (i) if $\ind_{i,j}=\True$, then $\chione{C}{A_i}{B_j} \leq \chi_- + \varepsilon_0 < \chi_{\min}$, and according to \autoref{ass:totalorder}, $A_i$ and $B_j$ are independent; (ii) if $\ind_{i,j}=\False$, then $\chione{C}{A_i}{B_j} > \chi_- - \varepsilon_0 > 0$, and according to \autoref{ass:totalorder}, $A_i$ and $B_j$ are correlated.
    In this case, \autoref{ass:totalorder} further ensures that $c_A(1),\dots,c_A(n)$ ($c_B(1),\dots,c_B(n)$) are distinct and \autoref{alg:order_linear} always gives the correct causal unravelling.

    The running time of \autoref{alg:order_linear} mostly attributes to \autoref{alg:ind}. We first discuss how to choose $N$, the number of queries to $\map{C}$. The relationship between $N$, $\varepsilon_0$ and $\kappa_0$ is given in \autoref{lem:ind_informal} as 
    $N = O\left(d_A^6 d_B^6 \varepsilon_0^{-2} \log(d_A d_B \kappa_0^{-1}) \right)$. 
    By our choice of parameters $\varepsilon_0 = \chi_{\min}/3$ and $\kappa_0 = \kappa / n^2 $, we have
    \begin{align}
        N = O\left(d_A^6 d_B^6 \chi_{\min}^{-2} \log(n d_A d_B \kappa^{-1}) \right). 
    \end{align}

    To analyze the computational complexity, we consider the two main parts of \autoref{alg:ind}: the query part and the estimation part. The query part invokes the channel $N$ times, and each time $n$ state preparation and $n$ measurements are performed. Here we assume that state preparation on a $d$-dimensional Hilbert space takes time $O(d)$ \cite{plesch2011quantum} and performing a POVM with $m$ outcomes takes time $O(m^2)$, according to the Naimark's dilation theorem that converts $m$-outcome POVM to projective measurement via an $m$-dimensional unitary gate \cite{peres2006quantum} implementable with $O(m^2)$ basic gates \cite{mottonen2004quantum}.  To perform the POVM with $m=d_B^2$ elements, we need $O(d_B^4)$ gates, and thus the query part takes a total time $O(Nn(d_A+d_B^4))$. The estimation part invokes \autoref{alg:c1} for $n^2$ times, and each call to \autoref{alg:c1} takes time $O(N+d_A^4d_B^4)$, $O(N)$ for computing $\hat{p}$ and $O(d_A^4d_B^4)$ for handling matrices and vectors. We don't count the time of inverting the matrix $F$ ($G$) since they are known before the execution of the algorithm and their inverses can be computed in advance. To sum up, the computational complexity of \autoref{alg:order_linear} is $O(Nn(n+d_A+d_B^4)+n^2d_A^4d_B^4)$. If we take $N=\Omega(d_A^4d_B^4)$, consistent with the analysis of the sample complexity in the previous paragraph, the first term dominates and the computational complexity is $O(Nn(n+d_A+d_B^4))$.
\end{proof}

\subsection{Details of the algorithm under Assumption \ref{ass:memoryless}}\label{app:memoryless}

In this section, we give the pseudocode of the algorithm under \autoref{ass:memoryless} and the proof of its correctness (\autoref{thm:memoryless}).

    \begin{algorithm}[H]
        \caption{Quantum causal unravelling algorithm under \autoref{ass:memoryless}\label{alg:memoryless}}
        \Indentp{0.5em}
        \SetInd{0.5em}{1em}
        \SetKwInOut{Input}{Input}
        \SetKwInOut{Output}{Output}
        \SetKw{Break}{break}
        \SetKwData{True}{true}
        \SetKwData{False}{false}
        \SetKwData{ind}{ind}
        \SetKwData{matched}{matched}
        \SetKwFunction{independent}{independence}
    
        \ResetInOut{Input}
        \Input{Black-box access of a quantum comb $\map{C}$ with input wires $A_1,\dots,A_n$ and output wires $B_1,\dots,B_n$, number of queries $N$, threshold $\chi_-$}
        \ResetInOut{Output}
        \Output{Ordering of inputs and outputs} 
        \BlankLine
    
        $\ind \gets \independent(\map{C}, N, \chi_-)$\;
        Initialize Boolean arrays $\matched_A$ and $\matched_B$ to $\False$\;
        \For{$i \gets 1$ \KwTo $n$}
        {
            \For{$j \gets 1$ \KwTo $n$}
            {
                \If{$\ind_{i,j} = \False$}
                {
                    $\pi(i) \gets j$\;
                    $\matched_A(i) \gets \True$\;
                    $\matched_B(j) \gets \True$\;
                    \Break\tcp*{Exit the loop for $j$}
                }
            }
        }\label{line:memoryless_match}
        \For{$i \gets 1$ \KwTo $n$}
        {
            \If{$\matched_A(i) = \False$}
            {
                \For{$j \gets 1$ \KwTo $n$}
                {
                    \If{$\matched_B(j) = \False$}
                    {
                        $\pi(i) \gets j$\;
                        $\matched_A(i) \gets \True$\;
                        $\matched_B(j) \gets \True$\;
                        \Break\tcp*{Exit the loop for $j$}
                    }
                }
            }
        }
        Output $(A_{1},B_{\pi(1)}),\dots,(A_{n},B_{\pi(n)})$\;
    \end{algorithm}
    \autoref{alg:memoryless} first matches the pairs $(A_i,B_j)$ that are correlated, and according to \autoref{ass:memoryless}, $A_i$ and $B_j$ must belong to the same tooth. For the inputs and outputs left unmatched by the first step, they can be matched arbitrarily, since the inputs and outputs have no correlation and are compatible with any causal unravelling.
    
    If we further assume that for each correlated input-output pair $(A_i,B_j)$, one has $\chione\rho{A_i}{B_j} \geq \chi_{\min}$ for some $\chi_{\min}>0$, then \autoref{alg:memoryless} produces an exact answer ($\map{C} \in \set{Comb}[ (A_1,B_{\pi(1)}),\dots,(A_n,B_{\pi(n)}) ]$) if one take $\chi_- = \chi_{\min}/3$, and the sample complexity is given by Eq. (\ref{eq:order_linear}). On the other hand, if such assumption of $\chi_{\min}$ does not exist, the algorithm still produces an answer, which is approximate. 

Here we consider the case where no assumption on $\chi_{\min}$ is made, and give the correctness and efficiency (\autoref{thm:memoryless}) merely based on \autoref{ass:memoryless}.

\begin{proof}[Proof of \autoref{thm:memoryless}]

    Picking $\chi_- = \varepsilon_0$, according to \autoref{lem:ind}, with probability $1-n^2 \kappa_0$,
    \begin{enumerate}
        \item if $\ind_{i,j} = \False$, then $\chione C{A_i}{B_j} > 0$
        \item if $\ind_{i,j} = \True$, then $\chione C{A_i}{B_j} \leq 2\chi_-$
    \end{enumerate}
    Since by \autoref{ass:memoryless}, for every input $A_i$, there is at most one output $B_j$ such that $\chione C{A_i}{B_j} > 0$, and therefore there is at most one $j$ satisfying $\ind_{i,j}=\False$. Let $A_{\rm matched} := \{A_i | \exists j, \ind_{i,j}=\False \}$ and $B_{\rm matched}:=\{ B_{\pi(i)} | A_i \in A_{\rm matched}\}$ be the inputs and outputs that are matched up to \autoref{line:memoryless_match} of \autoref{alg:memoryless}. 

    Define the channel $\map{D}$ by its action on product states as follows:
    \begin{align}
        \map{D}(\rho_{A_1} \otimes \dots \otimes \rho_{A_n}) :=  \bigotimes_{i \in A_{\rm matched}} \map{C}_{A_i \to B_{\pi(i) }}(\rho_{A_i}) \otimes   \bigotimes_{i \notin A_{\rm matched}}  \map{C}_{A_i \nrightarrow B_{\pi(i)} }(\rho_{A_i})
    \end{align}
    where $\map{C}_{A_i \to B_{\pi(i) }}(\rho_{A_i}) := \Tr_{B_{\neq \pi(i)}}[\map{C}(\rho_{A_i} \otimes I_{A_{\neq i}}/d_{A_{\neq i}})]$ is channel $\map{C}$ restricted on input $A_i$ and output $B_{\pi(i)}$, and $\map{C}_{A_i \nrightarrow B_{\pi(i)} }(\rho_{A_i}) := \Tr[\rho_{A_i}] C_{B_{\pi(i)}}$ is a constant channel that outputs $C_{B_{\pi(i)}}$, the marginal Choi state of $\map{C}$ on system $B_{\pi(i)}$.

    Clearly, $\map{D} \in \set{Comb}[ (A_{1}, B_{\pi(1)}), \dots , (A_{n}, B_{\pi(n)}) ]$, since it is a tensor product of channels from $A_{i}$ to $B_{\pi(i)}$. Now we show that the difference between $\map{C}$ and $\map{D}$ is small.

    According to \autoref{ass:memoryless}, $\map{C} = \bigotimes_{i=1}^n \map{C}_{A_i \to B_{\pi'(i)}} = \bigotimes_{i \in A_{\rm matched}} \map{C}_{A_i \to B_{\pi'(i) }} \otimes   \bigotimes_{i \notin A_{\rm matched}}  \map{C}_{A_i \to B_{\pi'(i) }}$.
    For channel $\map{C}$ and each $i \in A_{\rm matched}$, $A_{i}$ is correlated to $B_{\pi(i)}$ but not correlated to all other outputs, and by the way $\pi$ is computed, one must have $\pi(i)=\pi'(i)$. Therefore,
    \begin{align}
        \| \map{C} - \map{D} \|_\diamond & = \left\|\bigotimes_{i \in A_{\rm matched}} \map{C}_{A_i \to B_{\pi'(i) }} \otimes   \bigotimes_{i \notin A_{\rm matched}}  \map{C}_{A_i \to B_{\pi'(i) }} -  \bigotimes_{i \in A_{\rm matched}} \map{C}_{A_i \to B_{\pi(i) }} \otimes   \bigotimes_{i \notin A_{\rm matched}} \map{C}_{A_i \nrightarrow B_{\pi(i)} } \right\|_\diamond \\
        & = \left\|\bigotimes_{i \notin A_{\rm matched}}  \map{C}_{A_i \to B_{\pi'(i) }}  -  \bigotimes_{i \notin A_{\rm matched}} \map{C}_{A_i \nrightarrow B_{\pi(i)} }  \right\|_\diamond \label{eq:CABCAB}
    \end{align}
    and the difference between $\map{C}$ and $\map{D}$ only occurs in the second part involving $i \notin A_{\rm matched}$. Since
    \begin{align}
        \bigotimes_{i \notin A_{\rm matched}}  \Tr[\rho_{A_i}] C_{B_{\pi(i)}} = \prod_{i \notin A_{\rm matched}}  \Tr[\rho_{A_i}] \bigotimes_{j \notin B_{\rm matched}} C_{B_j} = \bigotimes_{i \notin A_{\rm matched}}  \Tr[\rho_{A_i}] C_{B_{\pi'(i)}}
    \end{align}
    we have $\bigotimes_{i \notin A_{\rm matched}}  \map{C}_{A_i \nrightarrow B_{\pi(i) }} = \bigotimes_{i \notin A_{\rm matched}}  \map{C}_{A_i \nrightarrow B_{\pi'(i) }}$, where $\map{C}_{A_i \nrightarrow B_{\pi'(i) }}$ is the constant channel defined as $\map{C}_{A_i \nrightarrow B_{\pi'(i)} }(\rho_{A_i}) := \Tr[\rho_{A_i}] C_{B_{\pi'(i)}}$. Eq. (\ref{eq:CABCAB}) then becomes
    \begin{align} \label{eq:CDdia}
        \| \map{C} - \map{D} \|_\diamond & = \left\|\bigotimes_{i \notin A_{\rm matched}}  \map{C}_{A_i \to B_{\pi'(i) }}  -  \bigotimes_{i \notin A_{\rm matched}} \map{C}_{A_i \nrightarrow B_{\pi'(i)} }  \right\|_\diamond \leq \sum_{i \notin A_{\rm matched}} \| \map{C}_{A_i \to B_{\pi'(i) }} - \map{C}_{A_i \nrightarrow B_{\pi'(i)} } \|_\diamond
    \end{align}


    Consider the channel $\map{C}_{A_i \to B_{\pi'(i) }}$. For $i \notin A_{\rm matched}$, $\ind_{i,\pi'(i)} = \True$ and with high probability $\chione C{A_i}{B_{\pi'(i)}} \leq 2\chi_-$. By definition of $\chi_1$,
    \begin{align}\label{eq:diff_choi_leq_chi}
        \| C_{A_i,B_{\pi'(i)}} - C_{A_i}\otimes C_{B_{\pi'(i)}} \|_1 = \left\| C_{A_i,B_{\pi'(i)}} - \frac{I_{A_i}}{d_{A_i}}\otimes C_{B_{\pi'(i)}} \right\|_1 \leq 2\chi_-
    \end{align}
    Eq. (\ref{eq:diff_choi_leq_chi}) is the distance between the Choi states of $\map{C}_{A_i \to B_{\pi'(i)} }$ and $\map{C}_{A_i \nrightarrow B_{\pi'(i)}}$, which gives a bound of the diamond-norm distance between the channels:
    \begin{align}\label{eq:CCdchi}
        \| \map{C}_{A_i \to B_{\pi'(i) }} - \map{C}_{A_i \nrightarrow B_{\pi'(i)} } \|_\diamond \leq d_{A_i} \left\| C_{A_i,B_{\pi'(i)}} - \frac{I_{A_i}}{d_{A_i}}\otimes C_{B_{\pi'(i)}} \right\|_1 \leq 2d_{A_i} \chi_-
    \end{align}

    Combining Eqs. (\ref{eq:CDdia}) and (\ref{eq:CCdchi}), we obtain
    \begin{align}
        \| \map{C} - \map{D} \|_\diamond \leq 2 \chi_- \sum_{i \notin A_{\rm matched}}d_{A_i} \leq 2 n d_A \varepsilon_0
    \end{align}
    where $d_A := \max_i d_{A_i}$. The sample complexity of \autoref{alg:memoryless} is given by \autoref{lem:ind_informal} as
    \begin{align}
        N = O\left(d_A^6 d_B^6 \varepsilon_0^{-2} \log(d_A d_B \kappa_0^{-1}) \right)
    \end{align}
    Let $\kappa := n^2 \kappa_0$ and $\varepsilon := 2 n d_A \varepsilon_0$, we have
    \begin{align}
        N = O\left(n^2 d_A^8 d_B^6 \varepsilon^{-2} \log(n d_A d_B \kappa^{-1}) \right)
    \end{align}
    and with probability $1-\kappa$,  $\| \map{C} - \map{D} \|_\diamond \leq \varepsilon$. The running time of \autoref{alg:memoryless} is the same as \autoref{alg:order_linear} and has been discussed in \autoref{app:totalorder}.
\end{proof}

\subsection{Analysis of the causal unravelling algorithm in the  approximate case}\label{app:approximate}

We first define the `approximate rank' of a state to characterize the lowest-rank state close to a given state with respect to the Hilbert-Schmidt distance.

\begin{definition}
    \begin{align}
        \rank_\eta(\rho) := \min_{\sigma: \|\sigma - \rho\|_2 \leq \eta} \rank(\sigma)
    \end{align}
\end{definition}

Using $\rank_\eta$, we write down the assumption for approximate causal unravelling, bounding the approximate ranks of marginal Choi state of $\map{C}$.

\begin{assumption} \label{ass:close_to_low_rank}
Let the output of \autoref{func:recursive} be $(A_{\sigma(1)},B_{\pi(1)}),\dots,(A_{\sigma(n)},B_{\pi(n)})$. 
Let $\map{C}^{(k,k-1)}$ be the reduced map of $\map{C}$ obtained by keeping only the fist $k$ inputs $A_{\sigma(1)},\dots,A_{\sigma(k)}$ and $k-1$ outputs $B_{\pi(1)},\dots,B_{\pi(k-1)}$, and let $C^{(k,k-1)}$ be the Choi state of $\map{C}^{(k,k-1)}$.
 We assume that $C^{(k,k-1)}$ are approximately polynomial-rank, namely
 \begin{align} \label{eq:rank_eta_leq_r}
    \forall k=2,\dots,n, \rank_{\eta^{(k)}}(C^{(k,k-1)}) = r^{(k)} \,, 
 \end{align}
where $\eta^{(k)}$ are small non-negative numbers and $r^{(k)}$ are bounded by a polynomial of $n$.
\end{assumption}

Note that \autoref{ass:close_to_low_rank} can be efficiently verified.
Ref. \cite{o2015quantum} provides an algorithm to decide whether a state $\rho$ is close to a low-rank state, namely whether there exists $\rho'$ such that $\|\rho-\rho'\|_1\leq\varepsilon$ and $\rank(\rho') \leq r$, in time $\Theta(r^2 \varepsilon^{-1})$. The algorithm in Ref. \cite{o2015quantum}  is for trace distance, but can be translated easily to the Hilbert-Schimidt distance via Eq. (\ref{eq:norm12_main}) to verify Eq. (\ref{eq:rank_eta_leq_r}).
Given this algorithm, we can verify \autoref{ass:close_to_low_rank} during the execution of \autoref{func:recursive}.
In the $(n-k+1)$-th recursive call, the value of $\map{C}$ entering the algorithm is $\map{C}^{(k)}$, and the state $C^{(k,k-1)}$ can be obtained by $C^{(k,k-1)} = \Tr_{B_y}[C^{(k)}]$, for every candidate $B_y$. Therefore, we can insert one extra step in the $\mathbf{foreach}$ loop to verify whether $C^{(k,k-1)}$ is close enough to a rank-$r^{(k)}$ state.
In the complexity analysis, we assume \autoref{ass:close_to_low_rank} is true and do not include the time to verify it. 

Under \autoref{ass:close_to_low_rank}, the error and efficiency of \autoref{func:recursive} are given in the following theorem.

\begin{theo}\label{thm:approximate}
Under \autoref{ass:close_to_low_rank}, for any confidence parameter $\kappa_0>0$, \autoref{func:recursive} outputs a causal unravelling $(A_1,B_{\pi(1)}),\dots,(A_n,B_{\pi(n)})$ satisfying the following conditions:

\begin{enumerate}
   \item With probability $1-\kappa_0$, the output of \autoref{func:recursive} is an approximate causal unravelling for $\map{C}$ in the following sense:
      \begin{align}
          \exists \map{D} \in \set{Comb}[ (A_{\sigma(1)},B_{\pi(1)}),\dots,(A_{\sigma(n)},B_{\pi(n)}) ], ~ \|C-D \|_1 \leq 8\sqrt{2} ~ n r_{\max}^{1/4} \eta_{\max}^{1/2}  \,,
      \end{align}
        where $\eta_{\max}:=\max_k \{\eta^{(k)}\}$ and $r_{\max}:=\max_k \{r^{(k)}\}$.
   \item The number of queries to $\map{C}$ is in the order of
    \begin{align}
        T_{\rm sample}=O \left(n^3\eta_{\max}^{-4} \log (n \kappa_0^{-1})\right) \,.
    \end{align}
   \item The computational complexity is in the order of $O(T_{\rm sample}n\log d_A)$.
\end{enumerate}
\end{theo}

We prove \autoref{thm:approximate} in the following subsections. In \autoref{appss:indep_approx}, we analyze the accuracy of the independence tests in the approximate case. In \autoref{appss:recursion_approx} we analyze how the error in each recursion step affects the final error of the algorithm. The main proof of \autoref{thm:approximate} is in \autoref{appss:main_proof_approx}.

\subsubsection{Accuracy of independence tests} \label{appss:indep_approx}

The following lemma is useful for converting the Hilbert-Schimidt distance to the trace distance in the approximate case.
\begin{lem} \label{lem:norm_soft_rank}
For two quantum states $\rho$ and $\sigma$, 
\begin{align}
    & \|\rho - \sigma \|_1 \leq  \sqrt{4\rank_\eta(\sigma)} \left( \|\rho - \sigma\|_2 +2\eta \right)
\end{align}
\end{lem}
\begin{proof}
By definition of $\rank_\eta$, there exists a state $\sigma'$ such that $\rank(\sigma')=\rank_\eta(\sigma)$ and $\|\sigma - \sigma'\|_2 \leq \eta$. By Eq. (\ref{eq:norm12_main}), $\|\rho - \sigma'\|_1\leq \sqrt{4\rank(\sigma')}\|\rho - \sigma'\|_2$, and we have
\begin{align}
    \|\rho - \sigma \|_1  & \leq \|\rho - \sigma'\|_1 + \|\sigma - \sigma'\|_1 \\
        & \leq \sqrt{4\rank(\sigma')}\|\rho - \sigma'\|_2  + \sqrt{4\rank(\sigma')}\|\sigma - \sigma'\|_2 \\
        & \leq \sqrt{4\rank(\sigma')} \left( \|\rho - \sigma\|_2 + \|\sigma - \sigma'\|_2 \right)  + \sqrt{4\rank(\sigma')}\|\sigma - \sigma'\|_2 \\
        & \leq \sqrt{4\rank(\sigma')} \left( \|\rho - \sigma\|_2 +2\|\sigma - \sigma'\|_2 \right) \\
        & = \sqrt{4\rank_\eta(\sigma)} \left( \|\rho - \sigma\|_2 +2\eta \right)
\end{align}
\end{proof}


We first consider the accuracy of the independence test implemented in \autoref{func:findlast_direct}, in an arbitrary iteration of \autoref{func:recursive}. Define $C_1 := \Tr_{B_y}[C]$ and $C_2 := I_{A_x}/d_{A_x} \otimes \Tr_{A_xB_y}[C]$ as in \autoref{func:findlast_direct}. By \autoref{lem:SWAP}, for each of the estimate $p_i, (i=1,2,3)$, with probability $1-\kappa$, $p_i$ is $\varepsilon$-close to the true value. By the union bound, with probability at least $1-3\kappa$, all three estimates are $\varepsilon$-close to their corresponding true values, and therefore
\begin{align} \label{eq:p1p2p3_dist_c}
 &~ \left|p_1+p_2-2p_3 - \|C_1- C_2\|_2^2 \right| \nonumber \\
\leq &~ \left|p_1- \Tr[C_1^2] \right| +\left|p_2- \Tr[C_2^2] \right| + \left|2p_3- 2\Tr[C_1C_2] \right| \nonumber \\
\leq &~ 4\varepsilon \,.
\end{align}

Next, by \autoref{lem:norm_soft_rank} and Eq. (\ref{eq:norm12_main}),
\begin{align} \label{eq:norm12_C1C2_c}
    \sqrt{2} \|C_1- C_2\|_2 \leq \|C_1- C_2\|_1 \leq \sqrt{4\rank_\eta(C_1)} ~ \left(\|C_1- C_2\|_2 + 2\eta\right)
\end{align}

Now we consider the output of \autoref{func:findlast_direct}. If \autoref{func:findlast_direct} returns \True, then $p_1+p_2-2p_3\leq \delta$. From Eq. (\ref{eq:p1p2p3_dist_c}), $\|C_1- C_2\|_2 \leq \sqrt{\delta + 4\varepsilon}$ and then from Eq. (\ref{eq:norm12_C1C2_c}), $\|C_1- C_2\|_1 \leq \sqrt{4\rank_\eta(C_1)} \left( \|C_1- C_2\|_2 + 2\eta \right) \leq \sqrt{4\rank_\eta(C_1)} \left( \sqrt{\delta + 4\varepsilon} + 2\eta\right)$. 
Since $\chione C {A_x}{A_{\neq x}B_{\neq y}} = \|C_1- C_2\|_1$, we have the following lemma:

\begin{lem}\label{lem:2norm_chi_direct_c}
Let $\map{C}$ be the input channel of \autoref{func:findlast_direct}. With probability $1-3\kappa$, \autoref{func:findlast_direct} satisfies the following:
\begin{enumerate}
\item If \autoref{func:findlast_direct} returns \True, then
\begin{align}
    \chione C {A_x}{A_{\neq x}B_{\neq y}} \leq \sqrt{4\rank_\eta(\Tr_{B_y}[C])} \left( \sqrt{\delta + 4\varepsilon} + 2\eta\right) \,, \label{eq:2norm_chi_c}
\end{align}
where $C$ is the Choi state of $\map{C}$;
\item If \autoref{func:findlast_direct} returns \False, then
\begin{align}
    \chione C {A_x}{A_{\neq x}B_{\neq y}} >\sqrt{2 (\delta - 4\varepsilon) }\,.
\end{align}
\end{enumerate}
\end{lem}

\subsubsection{Combining errors in all recursions}\label{appss:recursion_approx}

\begin{lem} \label{lem:p-tr}
For quantum states $\rho$ and $\sigma$ in Hilbert space $\spc{H}_A$, if $\|\rho - \sigma\|_1 \leq \varepsilon$, then for any purification of $\rho$, $\ket{\psi}_{AR}$ with an ancillary system $R$ satisfying $\dim R \geq \rank(\sigma)$, there exists a purification of $\sigma$, $\ket{\phi}_{AR}$ such that
\begin{align}
    \| \ketbra{\psi} - \ketbra{\phi} \|_1 \leq 2\sqrt\varepsilon
\end{align}
\end{lem}
\begin{proof}
By Uhlmann's Theorem, there exists a purification of $\sigma$, $\ket{\phi}_{AR}$ such that $F(\rho,\sigma) = F(\ketbra{\psi}, \ketbra{\phi})$.
\begin{align}
    \frac12 \| \ketbra{\psi} - \ketbra{\phi} \|_1 = \sqrt{1-F(\ketbra{\psi}, \ketbra{\phi})} = \sqrt{1-F(\rho,\sigma)} \leq \sqrt{1-\left(1-\frac12\|\rho - \sigma\|_1\right)^2} \leq \sqrt{\varepsilon - \varepsilon^2/4} \leq \sqrt{\varepsilon}
\end{align}
\end{proof}

\begin{lem}\label{lem:2ton}
Let $\map{C}$ be a quantum channel and $C$ be its Choi state. Let $C^{(i)}:= \Tr_{A_{i+1}B_{i+1}\dots A_{n}B_{n}}[C]$ be the Choi operator on first $i$ input-output pairs. For $n\geq 2$, if
\begin{align} \label{eq:CDi}
    \forall i=2,\dots,n,~ \exists \map{D}^{(i)\circ} \in \set{Comb}[(A_1\dots A_{i-1},B_1\dots B_{i-1}), (A_i, \emptyset )], ~ \|C^{(i)\circ} - D^{(i)\circ}\|_1 \leq \delta^{(i)} \,, 
\end{align}
(Note that the above can be directly obtained from $\chione{C}{A_i}{A_1\dots A_{i-1}B_1\dots B_{i-1}} \leq \delta^{(i)}$.) 
where $C^{(i)\circ}:= \Tr_{B_i}[C^{(i)}]$, then
\begin{align}
    \exists \map{D} \in \set{Comb}[(A_1,B_1), \dots , (A_n, B_n)], ~ \|C - D\|_1 \leq 2\sqrt{\delta^{(2)}} + 4 \sum_{i=3}^n \sqrt{\delta^{(i)}} \,,
\end{align}
where $D$ is the Choi state of $\map{D}$.
\end{lem}
\begin{proof}~

In this proof, calligraphic letters denote linear maps and italic letters denote the corresponding Choi states.

Define a proposition $P(i)$: there exists a purification of $C^{(i)}$, denoted as $C^{(i)*} \in S(\spc{H}_{A_1 \dots A_{i}} \otimes \spc{H}_{B_1 \dots B_{i}} \otimes \spc{H}_{M_i})$, where $M_i$ is the purifying system, and an isometric channel $\map E^{(i)*} \in \set{Comb}[(A_1,B_1),\dots, (A_{i-1}, B_{i-1}), (A_i, B_iM_i)]$ such that
\begin{align}
    \|C^{(i)*} - E^{(i)*} \|_1 \leq \varepsilon^{(i)} \,.
\end{align}

Consider $P(2)$. Pick any purification of $C^{(2)}$ as $C^{(2)*}$, and $C^{(2)*}$ is also a purification of $C^{(2)\circ}$. Pick a purification of $D^{(2)\circ}$ as $D^{(2)*}$ such that its trace distance to $C^{(2)*}$ is minimized. Since $\|C^{(2)\circ} - D^{(2)\circ}\|_1 \leq \delta^{(2)}$, according to \autoref{lem:p-tr}, we have $\|C^{(2)*} - D^{(2)*}\|_1 \leq 2\sqrt{\delta^{(2)}}$. By choosing $E^{(2)*}=D^{(2)*}$, we prove $P(2)$ with $\varepsilon^{(2)} = 2\sqrt{\delta^{(2)}}$.

Assume $P(i-1)$ holds with $\varepsilon^{(i-1)}$, then there exists a purification of $C^{(i-1)}$, denoted as $C^{(i-1)*} \in S(\spc{H}_{A_1 \dots A_{i-1}} \otimes \spc{H}_{B_1 \dots B_{i-1}} \otimes \spc{H}_{M_{i-1}})$, and an isometric channel $\map E^{(i-1)*} \in \set{Comb}[(A_1,B_1),\dots, (A_{i-2}, B_{i-2}), (A_{i-1}, B_{i-1}M_{i-1})]$ such that
\begin{align}\label{eq:CEminus1}
    \|C^{(i-1)*} - E^{(i-1)*} \|_1 \leq \varepsilon^{(i-1)} \,.
\end{align}

Pick a purification of $C^{(i)}$ as $C^{(i)*} \in S(\spc{H}_{A_1 \dots A_{i}} \otimes \spc{H}_{B_1 \dots B_{i}} \otimes \spc{H}_{M_{i}})$ where $\spc{H}_{M_{i}}$ is a large enough Hilbert space. Then $C^{(i)*}$ is also a purification of $C^{(i)\circ}$. Pick a purification of $D^{(i)\circ}$ as $D^{(i)*}$ such that its trace distance to $C^{(i)*}$ is minimized. Since $\|C^{(i)\circ} - D^{(i)\circ}\|_1 \leq \delta^{(i)}$ (\ref{eq:CDi}), according to \autoref{lem:p-tr}, we have
\begin{align} \label{eq:CDstar}
\| C^{(i)*} - D^{(i)*} \|_1 \leq 2\sqrt{\|C^{(i)\circ} - D^{(i)\circ}\|_1} \leq 2\sqrt{\delta^{(i)}} \,.
\end{align}

Let $(D')^{(i-1)} := \Tr_{A_i}[ D^{(i)\circ} ]$. Pick a purification of $(D')^{(i-1)}$ as $(D')^{(i-1)*}\in S(\spc{H}_{A_1 \dots A_{i-1}} \otimes \spc{H}_{B_1 \dots B_{i-1}} \otimes \spc{H}_{M_{i-1}})$  such that its trace distance to $C^{(i-1)*}$ is minimized. From \autoref{lem:p-tr} we have
\begin{align}\label{eq:DprimeC}
    \|(D')^{(i-1)*} - C^{(i-1)*}\|_1 \leq 2\sqrt{\|(D')^{(i-1)} - C^{(i-1)}\|_1 } = 2\sqrt{\| \Tr_{A_i}[D^{(i)\circ}-C^{(i)\circ}] \|_1 } \leq 2\sqrt{\| D^{(i)\circ} - C^{(i)\circ} \|_1 } \leq 2 \sqrt{\delta^{(i)}}
\end{align}

Note that $\map D^{(i)*} \in \set{Comb}[(A_1\dots A_{i-1},B_1\dots B_{i-1}), (A_i, B_iM_i)]$ since $\map{D}^{(i)\circ} \in \set{Comb}[(A_1\dots A_{i-1},B_1\dots B_{i-1}), (A_i, \emptyset )]$. Therefore $\map D^{(i)*}$ can be written as a concatenation of two channels with the first channel being $(D')^{(i-1)*}$:

\begin{align}
    \map D^{(i)*} =&~ (\map I_{B_1\dots B_{i-1}} \otimes \map{D}_i) \circ ((\map D')^{(i-1)*} \otimes \map I_{A_i})  \label{eq:Distar}     
\end{align}
where $\map D_i: S(\spc{H}_{A_i}\otimes\spc{H}_{M_{i-1}}) \to S(\spc{H}_{B_i}\otimes\spc{H}_{M_i})$ is an isometric channel since $D^{(i)*}$ is isometric.

Define
\begin{align}
    \map E^{(i)*} := (\map I_{B_1\dots B_{i-1}} \otimes \map{D}_i) \circ (\map E^{(i-1)*} \otimes \map I_{A_i})
\end{align}
and let $E^{(i)*}$ be its Choi state. This is an isometric channel because both $\map{D}_i$ and $\map E^{(i-1)*}$ are. Define $\tilde{\map{D}}_i(\rho): S(\spc{H}_{M_{i-1}}) \to S(\spc{H}_{A'_i}\otimes \spc{H}_{B_i} \otimes \spc{H}_{M_i})$ as $\tilde{\map{D}}_i(\rho) :=  (\map{D}_i \otimes \map I_{A'_i})( \rho \otimes \ketbra{\Phi^+}_{A_i A'_i})$, which is an isometric channel and preserves trace norm. Then we have
\begin{align}
\nonumber \| D^{(i)*} - E^{(i)*} \|_1 & = \left\| (\tilde{\map{D}}_i \otimes \map I_{A'_1 \dots A'_i} \otimes \map I_{B_1 \dots B_{i-1}}) \left( (D')^{(i-1)*} - E^{(i-1)*} \right) \right\|_1 \\
\nonumber    & = \| (D')^{(i-1)*} - E^{(i-1)*} \|_1 \\
\nonumber    & \leq \| (D')^{(i-1)*} - C^{(i-1)*} \|_1 + \| C^{(i-1)*} - E^{(i-1)*} \|_1 \\
    & \leq 2\sqrt{\delta^{(i)}} + \varepsilon^{(i-1)} \label{eq:DEstar}
\end{align}
where the last step uses Eqs. (\ref{eq:DprimeC}) and (\ref{eq:CEminus1}).

Last, from Eqs. (\ref{eq:CDstar}) and (\ref{eq:DEstar}),
\begin{align}
\| C^{(i)*} - E^{(i)*} \|_1 \leq \| C^{(i)*} - D^{(i)*} \|_1 + \| D^{(i)*} - E^{(i)*} \|_1 \leq 4\sqrt{\delta^{(i)}} + \varepsilon^{(i-1)}
\end{align}
Therefore, we can choose $\varepsilon^{(i)} = \varepsilon^{(i-1)} + 4\sqrt{\delta^{(i)}}$. Performing induction with base case $\varepsilon^{(2)} = 2\sqrt{\delta^{(2)}}$, we have
\begin{align}
    \varepsilon^{(n)} = 2\sqrt{\delta^{(2)}} + 4 \sum_{i=3}^n \sqrt{\delta^{(i)}}
\end{align}
Picking $D := \Tr_{M_n}[ E^{(n)*}]$, we have
\begin{align}
    \| C - D \|_1 \leq \| C^{(n)*} - E^{(n)*} \|_1 \leq \varepsilon^{(n)} = 2\sqrt{\delta^{(2)}} + 4 \sum_{i=3}^n \sqrt{\delta^{(i)}}
\end{align}
\end{proof}

\subsubsection{Proof of \autoref{thm:approximate}} \label{appss:main_proof_approx}

Now we are ready to prove \autoref{thm:approximate}.

\begin{proof}[Proof of \autoref{thm:approximate}]

In \autoref{func:recursive}, using the relation $T_{\rm test}(n) \leq n^2 + T_{\rm test}(n-1)$, we can find that there are $T_{\rm test}(n) \leq n^3$ independence tests, each of which invokes three SWAP tests. Each SWAP test produces an estimate within error $\varepsilon$ with probability no less than $1-\kappa$ according to \autoref{lem:SWAP}.
Let $\kappa = \kappa_0/3n^3$, and from now we assume all SWAP tests are within error $\varepsilon$, which has probability no less than $1-3n^3\kappa = 1-\kappa_0$.

Let $\map{C}^{(n)}$ be the initial value of $\map{C}$ in \autoref{func:recursive}, and let $\map{C}^{(k)}$, which has $k$ input wires and $k$ output wires, be the value of $\map{C}$ in the $(n-k+1)$-th recursive call of \autoref{func:recursive}.

Consider testing whether $(A_x,B_y)$ is the last tooth of $\map{C}^{(k)}$ with \autoref{func:findlast_direct}. Let $S$ be the set of all wires of $\map{C}^{(k)}$ excluding $A_x$ and $B_y$. If $(A_x,B_y)$ passes the test, namely \autoref{func:findlast_direct} returns \True, by \autoref{lem:2norm_chi_direct_c}, 
\begin{align}
    \chione {C^{(k)}}{A_x}{S} \leq \sqrt{4 r^{(k)}} \left( \sqrt{\delta + 4\varepsilon} + 2\eta^{(k)}\right) \,.
\end{align}

By \autoref{lem:2ton}, there exists $ \map{D} \in \set{Comb}[(A_{\sigma(1)},B_{\pi(1)}),\dots,(A_{\sigma(n)},B_{\pi(n)})]$ satisfying
\begin{align} \label{eq:CDleqrmaxetamax}
 \|C - D\|_1 \leq 
    4\sum_{i=2}^n \sqrt{\sqrt{4 r^{(k)}} \left( \sqrt{\delta + 4\varepsilon} + 2\eta^{(k)}\right)} \leq 4n \sqrt{\sqrt{4 r_{\max}} \left( \sqrt{\delta + 4\varepsilon} + 2\eta_{\max}\right)}\,,
\end{align}
where $\eta_{\max}:=\max_k \{\eta^{(k)}\}$ and $r_{\max}:=\max_k \{r^{(k)}\}$.


Now we pick $\delta$ and $\varepsilon$ such that $\sqrt{\delta+4\varepsilon}=2\eta_{\max}$ and $\delta=4\varepsilon$, which leads to $\delta=2\eta_{\max}^2$. Then Eq. (\ref{eq:CDleqrmaxetamax}) simplifies to
\begin{align}
 \|C - D\|_1 \leq 8\sqrt{2} ~ n r_{\max}^{1/4} \eta_{\max}^{1/2} 
\end{align}

According to \autoref{func:SWAPTEST}, $N=O(\varepsilon^{-2} \log\kappa^{-1}) = O(\delta^{-2} \log (n \kappa_0^{-1})) = O(\eta_{\max}^{-4} \log (n \kappa_0^{-1}))$. The sample complexity is then
\begin{align}
    T_{\rm sample} = N T_{\rm test}(n) = O(n^3\eta_{\max}^{-4} \log (n \kappa_0^{-1}))
\end{align}
and the computational complexity is $O(T_{\rm sample} n \log d_A)$ from the realization of SWAP tests.


\end{proof}

\subsection{A quantum causal unravelling algorithm with $c\geq 1$} \label{app:largerc}

Here we give the generalization of \autoref{func:recursive} for larger $c$, which decomposes the original process into the form of \autoref{fig:PQcomb} where each interaction involves at most $c$ inputs and $c$ outputs.

\begin{algorithm}[H]
    \caption{Quantum causal unravelling algorithm with $c\geq 1$}\label{func:recursive_c}
    \Indentp{0.5em}
    \SetInd{0.5em}{1em}
    \SetKwInOut{Preprocessing}{Preprocessing}
    \SetKwInOut{Input}{Input}
    \SetKwInOut{Output}{Output}
    \SetKw{Next}{next}
    \SetKwFunction{unravel}{unravel}

    \ResetInOut{Input}
    \Input{Black-box access to quantum channel $\map{C}$, upper bound $c$ of the partition size}
    \ResetInOut{Output}
    \Output{A causal unravelling  of $\map{C}$}
    \BlankLine
    \For{$(c_P, c_Q) \in \{1,\dots,c\}\times\{1,\dots,c\}$}
    {
       \ForEach{$P \subset \{A_1,\dots,A_{n_{\rm in}}\}$ with $|P|=c_P$ and $Q\subset \{B_1,\dots,B_{n_{\rm out}}\}$ with $|Q|=c_Q$}
       {
           \If(\tcp*[f]{Done by checking whether $P$ and $A_{\notin P}B_{\notin Q}$ are independent}){$(P,Q)$ is the last tooth}
           {
               Let $\map{C}_{A_{\notin P},B_{\notin Q}}(\rho) := \Tr_{Q}[\map{C}(\rho \otimes \tau_{P})]$, where $\tau_P$ is an arbitrary state on $P$\;
                \DontPrintSemicolon\tcp*[r]{Remove $P$ and $Q$ from $\map{C}$}\PrintSemicolon 
               Recursively run this algorithm on the reduced channel $\map{C}_{A_{\notin P},B_{\notin Q}}$\;
               Output $(P,Q)$ and exit this function\tcp*{Append $(P,Q)$ to the end of the output}
           }
       }
    }
    \tcp{The algorithm reaches here in two cases:}
    \tcp{1. The channel $\map{C}$ has no input wires or no output wires, in which case the foreach-loop is skipped; or}
    \tcp{2. The channel $\map{C}$ cannot be further decomposed by picking no more than $c$ inputs and outputs}
        Let $P$ be the set of input wires and $Q$ be the set of output wires\;
        Output $(P,Q)$ and exit this function\;
\end{algorithm}

The order that $c_P$ and $c_Q$ are enumerated is not explicit from the description above, but it is recommended to give priority to smaller values of $c_P$ and $c_Q$. This ensures that the algorithm exhausts all possibilities for smaller partition sizes before considering larger partitions, and will output the causal unravelling with minimal partition size.
This is preferred since smaller partition sizes unravels more causal information about the channel.

We analyze the complexity of \autoref{func:recursive_c} based on the following assumption similar to \autoref{ass:close_to_low_rank}.

\begin{assumption} \label{ass:close_to_low_rank_c}
Let the output of \autoref{func:recursive_c} be $(P_{1},Q_{1}),\dots,(P_m,Q_m)$. 
Let $\map{C}^{(k,k-1)}$ be the reduced map of $\map{C}$ obtained by keeping only the fist $k$ partitions of inputs $P_1\cup\dots\cup P_k$ and $k-1$ partitions of outputs $Q_1 \cup\dots\cup Q_{k-1}$, and let $C^{(k,k-1)}$ be the Choi state of $\map{C}^{(k,k-1)}$.
 We assume that $C^{(k,k-1)}$ are approximately polynomial-rank, namely
 \begin{align} \label{eq:rank_eta_leq_r_c}
    \forall k=2,\dots,m, \rank_{\eta^{(k)}}(C^{(k,k-1)}) = r^{(k)} \,, 
 \end{align}
where $\eta^{(k)}$ are small non-negative numbers and $r^{(k)}$ are bounded by a polynomial of $n$.
\end{assumption}

The complexity of \autoref{func:recursive_c} is given in the following theorem: 

\begin{theo}\label{thm:largerc}
Under \autoref{ass:close_to_low_rank_c}, for any confidence parameter $\kappa_0>0$, \autoref{func:recursive_c} outputs a causal unravelling $(P_{1},Q_{1}),\dots,(P_m,Q_m)$ satisfying the following conditions:

\begin{enumerate}
   \item With probability $1-\kappa_0$, the output of \autoref{func:recursive_c} is an approximate causal unravelling for $\map{C}$ in the following sense:
      \begin{align}
          \exists \map{D} \in \set{Comb}[ (P_{1},Q_{1}),\dots,(P_m,Q_m) ], ~ \|C - D \|_1 \leq 8\sqrt{2} ~ m r_{\max}^{1/4} \eta_{\max}^{1/2}  \,,
      \end{align}
        where $\eta_{\max}:=\max_k \{\eta^{(k)}\}$ and $r_{\max}:=\max_k \{r^{(k)}\}$.
   \item The number of queries to $\map{C}$ is in the order of
    \begin{align}
        T_{\rm sample}=O \left(n^{2c+1}\eta_{\max}^{-4} \log (n \kappa_0^{-1})\right) \,
    \end{align}
    where $n = \max\{n_{\rm in},n_{\rm out}\}$.
   \item The computational complexity is in the order of $O(T_{\rm sample}n\log d_A)$.
\end{enumerate}
\end{theo}

\begin{proof}
We first count the number of independence tests in \autoref{func:recursive_c}. Let $T_{\rm test}(n)$ be the number of independence tests required for a channel with $n_{\rm in}\leq n$ input wires and $n_{\rm out}\leq n$ output wires. The number of possibilities to choose from $n$ wires a subset of no more than $c$ wires  is $\sum_{k=1}^c \binom{n}{c} = O(n^c)$. In the worst case, we need to enumerate all such subsets of input and output wires, which takes time $O(n^{2c})$.
In each recursion, the number of inputs and outputs of the channel is reduced by at least one, so we obtain the recursive relation $T_{\rm test}(n) \leq O(n^{2c})+T_{\rm test}(n-1)$ with $T_{\rm test}(0)=O(1)$, solving which gives $T_{\rm test}(n)=O(n^{2c+1})$.

The remaining part of this proof is similar to the proof of \autoref{thm:approximate} in \autoref{appss:main_proof_approx}.

Let $\map{C}^{(n)}$ be the initial value of $\map{C}$ in \autoref{func:recursive}, and let $\map{C}^{(k)}$, which has $k$ partitions of inputs $P_1\cup\dots\cup P_k$ and $k$ partitions of outputs $Q_1 \cup\dots\cup Q_{k}$, be the value of $\map{C}$ in the $(m-k+1)$-th recursive call of \autoref{func:recursive}.

Consider testing whether $(P,Q)$ is the last tooth of $\map{C}^{(k)}$ with \autoref{func:findlast_direct}. Let $S$ be the set of all wires of $\map{C}^{(k)}$ excluding $P$ and $Q$. If $(P,Q)$ passes the test, namely \autoref{func:findlast_direct} returns \True, by \autoref{lem:2norm_chi_direct_c}, 
\begin{align}
    \chione {C^{(k)}}{P}{S} \leq \sqrt{4 r^{(k)}} \left( \sqrt{\delta + 4\varepsilon} + 2\eta^{(k)}\right) \,.
\end{align}

By a slight variation of \autoref{lem:2ton}, where the premise and conclusion are defined for the newly defined $\map{C}^{(k)}$ in this proof, there exists $ \map{D} \in \set{Comb}[(P_1,Q_1),\dots,(P_m,Q_m)]$ satisfying
\begin{align} \label{eq:CDleqrmaxetamax_c}
 \|C - D\|_1 \leq 4 m \sqrt{\sqrt{4 r_{\max}} \left( \sqrt{\delta + 4\varepsilon} + 2\eta_{\max}\right)}\,,
\end{align}
where $\eta_{\max}:=\max_k \{\eta^{(k)}\}$ and $r_{\max}:=\max_k \{r^{(k)}\}$. Note that the induction in the proof of \autoref{lem:2ton} now has $m$ steps instead of $n$.


Now we pick $\delta$ and $\varepsilon$ such that $\sqrt{\delta+4\varepsilon}=2\eta_{\max}$ and $\delta=4\varepsilon$, which leads to $\delta=2\eta_{\max}^2$. Then Eq. (\ref{eq:CDleqrmaxetamax_c}) simplifies to
\begin{align}
 \|C - D\|_1 \leq 8\sqrt{2} ~ n r_{\max}^{1/4} \eta_{\max}^{1/2} 
\end{align}

According to \autoref{func:SWAPTEST}, $N=O(\varepsilon^{-2} \log\kappa^{-1}) = O(\delta^{-2} \log (n \kappa_0^{-1})) = O(\eta_{\max}^{-4} \log (n \kappa_0^{-1}))$. The sample complexity is then
\begin{align}
    T_{\rm sample} = N T_{\rm test}(n) = O(n^3\eta_{\max}^{-4} \log (n \kappa_0^{-1}))
\end{align}
and the computational complexity is $O(T_{\rm sample} n \log d_A)$ from the realization of SWAP tests.

\end{proof}



\end{document}